\documentclass[11pt]{article}
\usepackage{amsmath,amsfonts,amsthm,graphicx,float,setspace,multirow}
\usepackage[usenames,dvipsnames]{color}
\usepackage[a4paper,margin=2.5cm]{geometry}

\newtheorem{theorem}{Theorem}
\newtheorem{lemma}{Lemma}
\newtheorem{remark}{Remark}

\newtheorem{prop}{Proposition}
\newtheorem{example}{Example}

\newenvironment{key}[1][]{ {\noindent \bf Keywords#1: } \rmfamily
}{\hspace*{\fill}  \\}

\begin{document}

\title{\LARGE Risk Aversion in the Small and in the Large: Beyond Arrow-Pratt \\ A Wiener Chaos Hierarchy of Dynamic Risk Premia}
\author{Christian Oliver Ewald\thanks{Adam Smith Business School, University of Glasgow, UK and Institute of Mathematics and Mathematical Statistics, Umeå University, Sweden. e-mail: christian.ewald@glasgow.ac.uk or christian.ewald@umu.se}}

\maketitle

\begin{abstract}
\noindent The Arrow--Pratt approximation is one of the cornerstones of expected utility theory, providing the classical local approximation of certainty equivalents and risk premia in terms of absolute risk aversion. Despite its widespread use, its mathematical scope and relationship to higher-order risk preferences remain only partially understood. This paper develops a new framework for the analysis of certainty equivalents and dynamic risk premia based on Malliavin calculus and Wiener chaos analysis. We first show that the classical Arrow--Pratt approximation is not asymptotically valid for arbitrary sequences of vanishing risks, thereby identifying precise limitations of the traditional theory. Motivated by this observation, we formulate certainty equivalents dynamically by considering the progressive revelation of uncertainty through a Brownian filtration. Combining Itô calculus, the Clark--Ocone representation and the Wiener chaos decomposition, we derive a complete hierarchy of higher-order dynamic risk premia and obtain explicit representations of the corresponding coefficients in terms of Malliavin derivatives. For mixed Wiener chaos expansions, higher-order preference measures, including prudence and temperance, emerge naturally through interactions between chaos components and are characterised using Bell polynomial representations. Explicit results for quadratic Gaussian functionals and the Vasicek interest-rate model illustrate the theory and identify a broad class of regular Wiener functionals for which the classical Arrow--Pratt approximation is recovered as the leading-order term. The results establish a unified framework linking expected utility theory, stochastic analysis and Wiener chaos expansions, opening a new perspective on higher-order certainty equivalents and the dynamic measurement of risk.
\end{abstract}

\begin{key}
Arrow–Pratt approximation; certainty equivalents; risk premium; Malliavin calculus; Wiener chaos expansions; Itô calculus; higher-order risk preferences; expected utility
\end{key}

\onehalfspacing

\section{Introduction}

The evaluation of risky prospects through certainty equivalents and risk
premia lies at the heart of decision theory, mathematical economics and
financial mathematics. Since the pioneering work of Arrow (1965) and
Pratt (1964), the local approximation

\[
\pi
\approx
\frac12
A(w)\sigma^2, \mbox{ with } A(w)
=
-\frac{u''(w)}{u'(w)},
\]

\noindent has become one of the most widely used tools for quantifying risk
attitudes. The Arrow--Pratt coefficient of absolute risk aversion, $A(w)$,
provides a remarkably simple characterization of local preferences under
uncertainty and has become a cornerstone of modern expected utility
theory.

The approximation has proved enormously influential. It underlies large
parts of the literature on portfolio choice, insurance, precautionary
saving, asset pricing, production under uncertainty and financial
economics. It also forms the basis for many textbook treatments of
expected utility, where it is commonly interpreted as the universal
first-order approximation to the certainty equivalent of a small risk.

During the last four decades, however, it has become increasingly clear
that many economic decisions cannot be adequately described by local
variance alone. Beginning with the seminal work of Kimball (1990), the
focus shifted towards higher-order risk preferences. Kimball introduced
the notion of prudence,
\[
P(w)
=
-
\frac{u'''(w)}{u''(w)},
\]
showing that the third derivative of utility governs precautionary
saving in the presence of future income risk. Prudence plays a role
analogous to that of Arrow--Pratt risk aversion, but for the demand for
self-insurance against future uncertainty rather than for the valuation
of a single risky prospect.

Subsequent work considerably broadened this perspective. Kimball (1992,
1993) introduced the concept of temperance,
\[
T(w)
=
-
\frac{u^{(4)}(w)}{u'''(w)},
\]
which measures the willingness to separate independent risks rather than
to concentrate them. Eeckhoudt and Schlesinger (2006), together with
Eeckhoudt, Gollier and Schlesinger (1996, 2005), developed a systematic
theory of higher-order risk attitudes by characterising prudence,
temperance and further higher-order preferences through simple lottery
comparisons. Their work demonstrated that successive derivatives of the
utility function admit natural economic interpretations extending the
Arrow--Pratt coefficient beyond local variance effects.

This higher-order preference literature has found numerous applications,
including precautionary saving, prevention, self-protection, portfolio
choice under background risk, production under uncertainty and optimal
insurance design. Nevertheless, the primary objective of this literature
has been to characterise comparative statics and behavioural properties
of preferences. Comparatively little attention has been paid to deriving
systematic higher-order approximations of certainty equivalents or risk
premia themselves. In particular, there exists no analogue of the
Arrow--Pratt approximation that naturally generates an infinite sequence
of higher-order dynamic risk-premium coefficients associated with these
successive preference measures.

A second strand of literature has focused on higher-order expansions of
expected utility and certainty equivalents. Building on the classical
mean--variance approximation, a number of authors have derived systematic
higher-order approximations based on cumulants, Edgeworth expansions and
Gram--Charlier series, thereby incorporating skewness, kurtosis and
higher moments into certainty-equivalent calculations. Such expansions
have found important applications in portfolio selection, asset pricing
and risk management, where non-Gaussian features of return distributions
play a significant role (see, for example, Athayde and Flores (2004),
Harvey and Siddique (2000), and Jondeau and Rockinger (2006)). These
approaches typically express the certainty equivalent as an expansion in
the cumulants of a fixed probability distribution, leading to correction
terms involving skewness, excess kurtosis and higher-order moments.
Consequently, higher-order preference measures such as prudence and
temperance naturally appear through successive derivatives of the utility
function.

While these contributions provide valuable insight into the effects of
distributional asymmetry and tail risk, they remain fundamentally static
in nature. The underlying source of uncertainty is fixed and the
expansion is performed with respect to the magnitude of higher moments or
departures from normality. By contrast, the present paper considers the
progressive revelation of uncertainty through a Brownian filtration and
studies the certainty equivalent as a function of the information horizon.
The resulting expansion is therefore not organised by cumulants of a
fixed distribution, but by the small-time evolution of the conditional
distribution. As we shall show, this leads naturally to a hierarchy
indexed by Wiener chaos order rather than by distributional moments, with
the coefficients admitting explicit representations in terms of
Malliavin derivatives and Wiener chaos contractions.

A third strand of literature has developed around stochastic analysis
and Malliavin calculus. Originating in the work of Malliavin (1978) and
subsequently developed by Bismut (1983), 
Ocone (1984) and
Nualart (2006), Malliavin calculus has become one of the principal tools
for analysing the differential structure of Wiener functionals. Its
central objects are the Malliavin derivative, which quantifies the
sensitivity of a functional with respect to infinitesimal perturbations
of the underlying Brownian path, and the Clark--Ocone representation,
which expresses square-integrable Wiener functionals as stochastic
integrals driven by their Malliavin derivatives. Together with the Wiener
chaos decomposition, these tools provide a complete orthogonal analysis
of Brownian functionals and have become fundamental in modern stochastic
analysis (Nualart, 2006; Di Nunno et al., 2009).

These methods have found numerous applications throughout mathematical
finance. Malliavin calculus has become a standard approach for computing
Greeks and sensitivity measures in derivative pricing, often providing
more efficient estimators than finite-difference methods (Fournié et al.,
1999, 2001). It also plays an important role in stochastic control,
utility maximisation and utility-indifference pricing (Karatzas and
Shreve, 1998; Henderson, 2002; Mania and Schweizer, 2005; Becherer,
2003; Monoyios, 2013), as well as in nonlinear filtering, backward
stochastic differential equations and numerical methods for high-
dimensional stochastic systems.

Closely related is the extensive literature exploiting Wiener chaos
expansions as an analytical and computational tool. Wiener chaos methods
have proved particularly successful in stochastic partial differential
equations, nonlinear filtering and uncertainty quantification, while in
mathematical finance they have been used for the approximation of
conditional expectations, the numerical solution of optimal stopping
problems and Bermudan option pricing (Bally, Pagès and Printems, 2005;
Bender and Denk, 2007; Bender and Steiner, 2012). Wiener chaos
expansions have also been employed in interest-rate theory, where they
provide explicit representations of square-integrable interest-rate
models and Wiener functionals associated with bond prices and the term
structure (Hughston and Rafailidis, 2005; Hughston and Mina, 2012).

Despite these developments, Malliavin calculus and Wiener chaos have
been employed almost exclusively as analytical or computational tools.
To the best of our knowledge, they have not been used to develop a
general theory of certainty equivalents and risk premia. In particular,
there appears to be no framework in which the classical Arrow--Pratt
approximation is systematically extended into a hierarchy of higher-order
risk premia indexed by Wiener chaos order, nor one in which successive
preference measures such as prudence and temperance emerge naturally
through Malliavin derivatives and Wiener chaos contractions. It is
precisely this connection that the present paper develops.

The present paper brings these three strands of literature together. We
develop a dynamic theory of higher-order Arrow--Pratt risk premia based
on Malliavin calculus and Wiener chaos expansions. Rather than viewing
risk as a static random variable whose magnitude is artificially scaled,
we study the gradual arrival of uncertainty through the Brownian
filtration. The certainty equivalent therefore becomes a function of the
information horizon, and its local behaviour admits a systematic
small-time expansion.

Our starting point is a careful re-examination of the classical
Arrow--Pratt approximation. We show that, contrary to widespread belief,
the approximation is not asymptotically correct for arbitrary sequences
of small risks. In particular, we construct sequences of risks that
converge to zero in $L^p$ for every $p$ while the relative error of the
Arrow--Pratt approximation remains bounded away from zero. This
demonstrates that small variance alone is insufficient to justify the
classical approximation and that additional structural assumptions are
required.

Motivated by this observation, we place the problem into a continuous
time Brownian framework. Using It\^o calculus we derive a dynamic
interpretation of the Arrow--Pratt approximation as an infinitesimal
risk premium generated by the local quadratic variation of the
conditional expectation process. This formulation naturally leads to a
representation in terms of the Clark--Ocone integrand and hence the
first Malliavin derivative of the payoff.

The main contribution of the paper is the construction of a complete
hierarchy of higher-order dynamic risk premia. Exploiting the Wiener
chaos decomposition, we prove that each pure Wiener chaos contributes
first at a deterministic order in the small-time expansion of the
certainty equivalent. The resulting hierarchy yields a sequence of
dynamic Arrow--Pratt coefficients indexed by the Wiener chaos order and
expressed explicitly through Malliavin derivatives.

The analysis is then extended to finite mixed Wiener chaos expansions.
In this setting interactions between different chaos orders generate
higher-order preference characteristics. In particular, prudence,
temperance and higher-order utility derivatives arise naturally through
contractions between Wiener chaos components, while Bell-polynomial
representations describe the nonlinear transformation from expected
utility to certainty equivalents. This provides, to our knowledge, the
first systematic representation of higher-order dynamic risk premia in
terms of Wiener chaos contractions and Malliavin derivatives.

The paper concludes with explicit examples based on quadratic Gaussian
payoffs and the Vasicek interest-rate model. These examples demonstrate
that, within the class of regular Wiener functionals, the classical
Arrow--Pratt approximation is recovered as the leading-order term while
the higher-order corrections are completely determined by the underlying
Wiener chaos structure. This stands in sharp contrast to the
counterexamples constructed earlier in the paper and allows us to
identify precisely the class of stochastic models for which the
Arrow--Pratt approximation is asymptotically valid.

The remainder of the paper is organised as follows. Section~2 reviews the classical Arrow--Pratt approximation and discusses its mathematical limitations. Section~3 introduces a dynamic formulation based on Itô calculus, conditional certainty equivalents and Malliavin calculus. Section~4 develops a hierarchy of dynamic risk premia for pure Wiener chaos expansions, showing that each chaos order contributes first at a deterministic order in the small-time expansion of the certainty equivalent and deriving explicit representations in terms of Malliavin derivatives. Section~5 extends the analysis to mixed Wiener chaos expansions, where interactions between different chaos orders give rise naturally to higher-order preference measures such as prudence and temperance, leading to a unified representation based on Bell polynomials. Finally, Section~6 establishes the asymptotic validity of the classical Arrow--Pratt approximation for a broad class of regular Wiener functionals and illustrates the theory through explicit examples.

\section{Classical Arrow--Pratt Approximation and its Limitations}

Before developing our dynamic framework, it is useful to revisit the
classical Arrow--Pratt approximation in some detail. Besides establishing
notation and recalling the standard derivation, this also allows us to
identify several mathematical issues that have received relatively little
attention in the literature. In particular, we shall show that the
Arrow--Pratt approximation is not asymptotically valid for arbitrary
sequences of vanishing risks, thereby motivating the need for a more
general framework developed in the subsequent sections.

The risk premium associated with a random payoff $w+\epsilon$, where
$w\in\mathbb{R}$ denotes the initial wealth level and
$\mathbb{E}(\epsilon)=0$, is defined as the amount of certain wealth that
an investor is willing to forgo in order to eliminate the risk. If the
investor's preferences are represented by a strictly increasing and concave utility
function $u(\cdot)$, then the risk premium $\pi$ is determined implicitly
through the certainty-equivalent equation
\begin{equation}
\mathbb{E}\bigl(u(w+\epsilon)\bigr)=u(w-\pi).\label{eq:1}
\end{equation}

At this stage it is useful to emphasise that the risk premium depends on
three separate ingredients: the initial wealth level $w$, the utility
function $u(\cdot)$ describing the investor's preferences, and the
distribution of the random disturbance $\epsilon$. To make this dependence
explicit, we may write
\[
\pi=\pi(w,u(\cdot),\epsilon).
\]
Except in a few special cases, the certainty-equivalent equation cannot be
solved explicitly, so that the corresponding risk premium is generally not
available in closed form. As a consequence, much of the theoretical
analysis of risk premia relies on asymptotic approximations. By far the
most influential of these is the local approximation developed
independently by Arrow (1965) and Pratt (1964), now universally known as
the Arrow--Pratt approximation. Owing to its simplicity and economic
interpretability, it has become one of the fundamental results of expected
utility theory and remains the standard tool for approximating certainty
equivalents under small risks.

Pratt (1964) derives this approximation by expanding both sides of the
certainty-equivalent equation in a Taylor series about the initial wealth
level and then solving the resulting expression for the risk premium. More
precisely, he considers the expansions
\begin{eqnarray}
u(w-\pi) & = &  u(x) - \pi u^\prime(w) + \mathcal{O}(\pi^2) \label{eq:2} \\
\mathbb{E}(u(w + \epsilon)) & = &  \mathbb{E}\left(u(x) + \epsilon u^\prime(w) + \frac{1}{2} \epsilon^2 u^{\prime \prime}(w) + \mathcal{O}(\epsilon^3)\right) \notag \\
& = & u(w) + \frac{1}{2} \sigma_\epsilon^2 u^{\prime \prime}(w) + \mathcal{O}(\sigma_\epsilon^2).\label{eq:3}
\end{eqnarray}
Ignoring the higher order $\mathcal{O}(\pi^2)$ terms, he then derives from (\ref{eq:2}) and (\ref{eq:3}) that
\begin{equation}
\pi(w,u(\cdot),\epsilon) = \frac{1}{2} \sigma_\epsilon^2 A(w), \label{eq:4}
\end{equation}
with $A(w) = -\frac{u^{\prime \prime}(w)}{u^\prime(w)}$ the Arrow-Pratt risk aversion coefficient.


Although the Arrow--Pratt approximation has become one of the most widely
used results in expected utility theory, the original derivation raises a
number of mathematical issues. The first concerns a key assumption made by
Pratt (1964) in passing from equation (\ref{eq:3}) to equation (\ref{eq:4}). Specifically,
he assumes that (quote)

\emph{``We assume the third absolute central moment of $\tilde{z}$ is of
smaller order than $\sigma_z^2$. (Ordinarily it is of order
$\sigma_z^3$.)''}

In our notation, $\tilde{z}$ corresponds to $\epsilon$. The remark in
parentheses is particularly noteworthy, as it acknowledges that the
assumption is not expected to hold in the situations most commonly
encountered in practice. Indeed, it already fails for the benchmark case
of a normally distributed random variable,
$\epsilon\sim\mathcal{N}(0,\sigma_\epsilon^2)$, for which
\[
\mathbb{E}(|\epsilon|^3)=\sqrt{\frac{2}{\pi}}\sigma_\epsilon^3,
\]
so that the third absolute central moment is of order
$\sigma_\epsilon^3$, rather than of smaller order than
$\sigma_\epsilon^2$.

A second issue concerns the structure of the derivation itself. The Taylor
expansion on the left-hand side is performed with respect to the random
variable $\epsilon$ and is retained up to second order, whereas the
expansion on the right-hand side is performed with respect to the unknown
risk premium $\pi$ and is truncated after the linear term. Since the two
expansions involve different variables and are carried out to different
orders, there is no direct asymptotic relationship between the remainder
terms $\mathcal{O}(\pi^2)$ and $\mathcal{O}(\epsilon^3)$ (or,
equivalently, $\mathcal{O}(\sigma_\epsilon^3)$). Consequently, the
derivation alone does not justify neglecting one remainder relative to the
other, and additional arguments are required to establish the validity of
the resulting approximation.

The following Proposition will show the limits of the approximation.

 



\begin{prop}\label{prop:arrow-pratt-counterexample}
Let \(u:I\rightarrow\mathbb{R}\) be a strictly increasing and strictly concave
utility function, where \(I\subseteq\mathbb{R}\) is an open interval. Suppose
that \(u\in C^{2}(I)\), let \(w\in I\), and assume that $u'(w)>0$
and $u''(w)<0$. Choose \(a>0\) such that \(w-a\in I\). Let \((p_n)_{n\geq 1}\) be a sequence satisfying
\[
0<p_n<1
\qquad\text{and}\qquad
\lim_{n\rightarrow\infty}p_n=0,
\]
and define
\begin{equation*}
b_n:=\frac{ap_n}{1-p_n}.
\end{equation*}
Consider the sequence of two-point risks
\begin{equation*}
\varepsilon_n=
\begin{cases}
-a, & \text{with probability } p_n,\\[1mm]
b_n, & \text{with probability } 1-p_n.
\end{cases}
\end{equation*}

Then the following statements hold:
\begin{enumerate}
\item $\mathbb{E}(\varepsilon_n)=0$ for every \(n\geq 1\).
\item For every \(k\geq 1\), $\lim_{n\rightarrow\infty}
\mathbb{E}\left(\lvert\varepsilon_n\rvert^k\right)=0$.

\item The variance satisfies
\begin{equation*}
\sigma_n^2
:=
\mathbb{E}\left(\varepsilon_n^2\right)
=
\frac{a^2p_n}{1-p_n}
=
a^2p_n+O(p_n^2).
\end{equation*}

\item Let \(\pi_n=\pi(w,u(\cdot),\varepsilon_n)\) denote the risk premium defined
implicitly by
\begin{equation*}
\mathbb{E}\left[u(w+\varepsilon_n)\right]
=
u(w-\pi_n).
\end{equation*}
Then
\begin{equation*}
\pi_n
=
\frac{
u(w)-u(w-a)-au'(w)
}{
u'(w)
}
p_n
+
O(p_n^2).
\end{equation*}

\item Consequently,
\begin{equation*}
\lim_{n\rightarrow\infty}
\frac{\pi_n}{\sigma_n^2}
=
\frac{
u(w)-u(w-a)-au'(w)
}{
a^2u'(w)
}.
\end{equation*}
\end{enumerate}

The Arrow--Pratt approximation is
\begin{equation*}
\pi_n^{\mathrm{AP}}
=
\frac{1}{2}A(w)\sigma_n^2
=
-\frac{1}{2}
\frac{u''(w)}{u'(w)}
\sigma_n^2.
\end{equation*}
Hence, whenever
\begin{equation}
u(w)-u(w-a)-au'(w)
\neq
-\frac{1}{2}a^2u''(w), \label{eq:12}
\end{equation}
one has
\begin{equation*}
\lim_{n\rightarrow\infty}
\frac{\pi_n}{\pi_n^{\mathrm{AP}}}
=
\frac{
2\left[u(w)-u(w-a)-au'(w)\right]
}{
-a^2u''(w)
}
\neq 1.
\end{equation*}

Thus, even though \(\varepsilon_n\rightarrow 0\) in the space \(L^k\) for every
\(k\geq 1\), the relative error of the Arrow--Pratt approximation need not
converge to zero.
\end{prop}

\begin{proof}
By construction,
\begin{align*}
\mathbb{E}(\varepsilon_n)
&=
-ap_n+(1-p_n)b_n  =
-ap_n+(1-p_n)\frac{ap_n}{1-p_n}
=
0.
\end{align*}

For every \(k\geq 1\),
\begin{align*}
\mathbb{E}\left(\lvert\varepsilon_n\rvert^k\right)
&=
a^kp_n+(1-p_n)b_n^k  =
a^kp_n
+
(1-p_n)
\left(
\frac{ap_n}{1-p_n}
\right)^k  =
a^k
\left(
p_n+
\frac{p_n^k}{(1-p_n)^{k-1}}
\right).
\end{align*}
Since \(p_n\to0\), both terms on the right-hand side converge to zero. Hence, $\mathbb{E}\left(\lvert\varepsilon_n\rvert^k\right) \longrightarrow0$ for every \(k\geq1\). In particular,
\begin{align*}
\sigma_n^2
&=
\mathbb{E}(\varepsilon_n^2)  =
a^2p_n+(1-p_n)b_n^2  =
a^2p_n+\frac{a^2p_n^2}{1-p_n}
=
\frac{a^2p_n}{1-p_n}.
\end{align*}
Since
\[
\frac{1}{1-p_n}
=
1+p_n+O(p_n^2),
\]
it follows that
\begin{equation}
\sigma_n^2
=
a^2p_n+O(p_n^2).
\label{eq:variance-expansion}
\end{equation}

We next compute expected utility. Since \(b_n=O(p_n)\), a Taylor
expansion of \(u\) around \(w\) yields
\[
u(w+b_n)
=
u(w)+u'(w)b_n+O(p_n^2).
\]
Therefore,
\begin{align*}
\mathbb{E}\left[u(w+\varepsilon_n)\right]
&=
p_nu(w-a)
+
(1-p_n)u(w+b_n) \\
&=
p_nu(w-a)
+
(1-p_n)
\left(
u(w)+u'(w)b_n+O(p_n^2)
\right) \\
&=
u(w)
-
\left(
u(w)-u(w-a)-au'(w)
\right)p_n
+
O(p_n^2),
\end{align*}
where we have used $(1-p_n)b_n=ap_n$. By strict concavity, $u(w-a)<u(w)-au'(w)$, and hence
\[
u(w)-u(w-a)-au'(w)>0.
\]

By definition of the risk premium,
\begin{equation}
u(w-\pi_n)
=
u(w)
-
\left(
u(w)-u(w-a)-au'(w)
\right)p_n
+
O(p_n^2).
\label{eq:expected-utility-expansion}
\end{equation}

Since the right-hand side of
\eqref{eq:expected-utility-expansion} converges to \(u(w)\), strict
monotonicity and continuity of \(u\) imply that $\pi_n\longrightarrow0$. Moreover, since \(u'(w)>0\), the inverse function theorem implies that
\(u\) possesses a twice continuously differentiable local inverse
\(v=u^{-1}\) in a neighbourhood of \(u(w)\). Applying \(v\) to
\eqref{eq:expected-utility-expansion} gives
\begin{align*}
w-\pi_n
&=
v\Bigl(
u(w)
-
\left(
u(w)-u(w-a)-au'(w)
\right)p_n
+
O(p_n^2)
\Bigr).
\end{align*}

A Taylor expansion of \(v\) around \(u(w)\) yields
\begin{align*}
w-\pi_n
&=
v(u(w))
-
v'(u(w))
\left(
u(w)-u(w-a)-au'(w)
\right)p_n
+
O(p_n^2).
\end{align*}
Since $v(u(w))=w$ and $v'(u(w)) = \frac{1}{u'(w)}$, we obtain
\begin{equation}
\pi_n
=
\frac{
u(w)-u(w-a)-au'(w)
}{
u'(w)
}
p_n
+
O(p_n^2).
\label{eq:risk-premium-expansion}
\end{equation}

Combining
\eqref{eq:risk-premium-expansion}
with
\eqref{eq:variance-expansion}
gives
\begin{equation}
\lim_{n\to\infty}
\frac{\pi_n}{\sigma_n^2}
=
\frac{
u(w)-u(w-a)-au'(w)
}{
a^2u'(w)
}.
\label{eq:risk-premium-variance-limit}
\end{equation}

Finally, the Arrow--Pratt approximation is
\[
\pi_n^{\mathrm{AP}}
=
-\frac{1}{2}
\frac{u''(w)}{u'(w)}
\sigma_n^2.
\]
Hence, using
\eqref{eq:risk-premium-variance-limit} and (\ref{eq:12}),
\begin{align*}
\lim_{n\to\infty}
\frac{\pi_n}{\pi_n^{\mathrm{AP}}}
&=
\frac{\left(
\displaystyle
\frac{
u(w)-u(w-a)-au'(w)
}{
a^2u'(w)
}
\right)}{
\displaystyle
-\frac{1}{2}
\left(\frac{u''(w)}{u'(w)}\right)
}  =
\frac{
2\left(
u(w)-u(w-a)-au'(w)
\right)
}{
-a^2u''(w)
} \neq 1.
\end{align*}
\end{proof}

\begin{example}
Consider exponential utility, $u(x)=-\exp(-\gamma x)$, with $\gamma>0$. Then $u'(x)=\gamma\exp(-\gamma x)$ and
$u''(x)=-\gamma^2\exp(-\gamma x)$, so the Arrow--Pratt coefficient of absolute risk aversion is constant:
\[
A(w)
=
-\frac{u''(w)}{u'(w)}
=
\gamma,
\]

which is well known. Moreover,
\begin{align*}
u(w)-u(w-a)-au'(w)
&=
-\exp(-\gamma w)
+\exp\bigl(-\gamma(w-a)\bigr)
-a\gamma\exp(-\gamma w) \\
&=
\exp(-\gamma w)
\left(
\exp(\gamma a)-1-\gamma a
\right).
\end{align*}
Substituting this expression into
\eqref{eq:risk-premium-variance-limit} gives
\begin{equation}
\lim_{n\to\infty}
\frac{\pi_n}{\sigma_n^2}
=
\frac{
\exp(\gamma a)-1-\gamma a
}{
\gamma a^2
}.
\label{eq:exponential-premium-variance-limit}
\end{equation}

Since $\pi_n^{\mathrm{AP}} = \frac{\gamma}{2}\sigma_n^2$, it follows directly from
\eqref{eq:exponential-premium-variance-limit} that
\begin{equation*}
\lim_{n\to\infty}
\frac{\pi_n}{\pi_n^{\mathrm{AP}}}
=
\frac{
2\left(
\exp(\gamma a)-1-\gamma a
\right)
}{
\gamma^2a^2
}.
\label{eq:exponential-relative-error}
\end{equation*}

For every \(a>0\),
\[
\exp(\gamma a)
>
1+\gamma a+\frac{1}{2}\gamma^2a^2.
\]
Consequently,
\[
\frac{
2\left(
\exp(\gamma a)-1-\gamma a
\right)
}{
\gamma^2a^2
}
>1.
\]
Hence,
\[
\lim_{n\to\infty}
\frac{\pi_n}{\pi_n^{\mathrm{AP}}}
>1.
\]
In conclusion, the Arrow--Pratt approximation asymptotically underestimates the true
risk premium.

For completeness, the exact risk premium can also be obtained explicitly.
The defining equation
\[
\mathbb{E}\left[u(w+\varepsilon_n)\right]
=
u(w-\pi_n)
\]
becomes
\[
p_n\exp\bigl(-\gamma(w-a)\bigr)
+
(1-p_n)\exp\bigl(-\gamma(w+b_n)\bigr)
=
\exp\bigl(-\gamma(w-\pi_n)\bigr).
\]
After cancelling the common factor \(\exp(-\gamma w)\), we obtain
\begin{equation*}
\pi_n
=
\frac{1}{\gamma}
\log\left[
p_n\exp(\gamma a)
+
(1-p_n)\exp(-\gamma b_n)
\right].
\label{eq:exact-exponential-risk-premium}
\end{equation*}

Using $b_n=\frac{ap_n}{1-p_n}$ and expanding the logarithm around \(p_n=0\) yields
\[
\pi_n
=
\frac{
\exp(\gamma a)-1-\gamma a
}{
\gamma
}
p_n
+
O(p_n^2).
\]
By contrast,
\[
\pi_n^{\mathrm{AP}}
=
\frac{\gamma a^2}{2}p_n
+
O(p_n^2).
\]
Thus the discrepancy between the exact and Arrow--Pratt coefficients
persists at first order in \(p_n\).
\end{example}

Proposition~\ref{prop:arrow-pratt-counterexample} demonstrates that the classical derivation of the Arrow--Pratt approximation cannot be justified solely on the basis of the Taylor expansions employed by Pratt (1964). This, however, should not be interpreted as evidence against the approximation itself. Rather, it indicates that its asymptotic validity requires a more careful mathematical justification.

Indeed, as pointed out by Gollier (2001, pp.~47--49), the Arrow--Pratt approximation can be derived without relying on the asymmetric Taylor expansions used by Pratt (1964). Instead, Gollier considers a family of scaled risks of the form $k\epsilon$ and studies the resulting risk premium as a function of the scaling parameter $k$. The central idea is that a satisfactory approximation should accurately describe the behaviour of the risk premium for sufficiently small risks and, in particular, become asymptotically exact as $k\to0$.  Following Gollier (2001), we therefore define
\[
\pi(k)=\pi\bigl(w,u(\cdot),k\epsilon\bigr),
\]
where, by a slight abuse of notation, $\pi(k)$ denotes the risk premium associated with the scaled risk $k\epsilon$. The defining certainty-equivalent equation (1) then becomes
\begin{equation}
\mathbb{E}(u(w + k \epsilon)) = u(w- \pi(k)). \label{eq:21}
\end{equation}
Equation (\ref{eq:21}) holds for all $k \geq 0$ and we consider its derivative with respect to $k$, assuming implicitly that the function $\pi(k)$ is differentiable with respect to $k$ and that expectation and differentiation can be interchanged.\footnote{Mathematical conditions for this can be provided.} This gives
\begin{equation}
\mathbb{E}(u^\prime(w + k \epsilon) \epsilon) = - u^\prime(w- \pi(k)) \pi^\prime(k). \label{eq:22}
\end{equation}
Differentiating both side of equation (\ref{eq:22}) once more we obtain
\begin{equation}
\mathbb{E}(u^{\prime\prime}(w + k \epsilon) \epsilon^2) = - u^{\prime\prime}(w- \pi(k)) (\pi^\prime(k))^2 - u^\prime(w- \pi(k)) \pi^{\prime \prime}(k). \label{eq:23}
\end{equation}
Evaluating equation (\ref{eq:22}) at $k=0$ gives 
\begin{equation}
\pi^\prime(0)=0.\label{eq:24}
\end{equation}
Evaluating equation (\ref{eq:23}) at $k=0$ and substitution of (\ref{eq:24}) gives
\begin{equation}
\pi^{\prime \prime}(0) = - \frac{u^{\prime\prime}(w)}{u^\prime(w)} \sigma^2,\label{eq:25}
\end{equation}
with $\sigma^2 = \mathbb{E}(\epsilon^2)$. Equation (\ref{eq:25}) is exact and represents a statement about the infinitesimal level of risk aversion at $k=0$. Gollier (2001) now continues to Taylor expand the function $\pi(k)$ up to order 2, i.e.
\begin{equation*}
\pi(k) = \pi(0) + \pi^\prime(0) k + \frac{1}{2} \pi^{\prime \prime}(0) + \mathcal{O}^3
\end{equation*}
and by substitution of (\ref{eq:24}) and (\ref{eq:25}) obtains
\begin{equation*}
\pi(w,u(\cdot),k\epsilon) \approx \frac{1}{2} \sigma^2 A(w),
\end{equation*}
with $A(w)$ the Arrow-Pratt coefficient of risk aversion $A(w) = - \frac{u^{\prime\prime}(w)}{u^{\prime}(w)}$.


Higher-order approximations may, in principle, be obtained by continuing this procedure, that is, by repeatedly differentiating equation (\ref{eq:21}) with respect to the scaling parameter $k$ and expanding the function $\pi(k)$ to successively higher orders. This leads to increasingly involved expressions involving higher derivatives of the utility function and the moments of the underlying risk.

While mathematically sound, this approach remains tied to the artificial scaling of a fixed random variable. The parameter $k$ does not correspond to a natural source of uncertainty, but is introduced solely to facilitate an asymptotic analysis. In many applications, particularly in continuous-time finance, uncertainty instead unfolds progressively as information is revealed over time. This suggests replacing the amplitude scaling of the risk by a dynamic formulation in which the uncertainty evolves through a stochastic process and the certainty equivalent becomes a function of the information horizon.

The next section develops precisely such a formulation. By replacing the scaling parameter with time and exploiting Itô calculus, we obtain a dynamic interpretation of the Arrow--Pratt approximation as an infinitesimal risk premium. As will become apparent, this perspective not only reproduces the classical Arrow--Pratt coefficient, but also provides the natural starting point for the hierarchy of higher-order dynamic risk premia developed later using Malliavin calculus and Wiener chaos expansions.

\section{A Dynamic Arrow–Pratt Framework via Itô and Malliavin Calculus}

The classical Arrow--Pratt approximation and Gollier's refinement both analyse the risk premium by artificially scaling the magnitude of a fixed random payoff. While this provides a mathematically convenient asymptotic framework, the scaling parameter itself has no intrinsic economic interpretation. In many applications, particularly in continuous-time finance, uncertainty is not introduced by multiplying a fixed risk by an external parameter but rather unfolds progressively as new information becomes available over time. This naturally suggests replacing amplitude scaling by a dynamic formulation in which the certainty equivalent evolves with the underlying flow of information.

The purpose of this section is to develop precisely such a dynamic framework. We first show that the classical Arrow--Pratt approximation admits a natural interpretation as an infinitesimal risk premium generated by the local evolution of uncertainty in a Brownian setting. This interpretation is then extended from Brownian motion to general square-integrable Wiener functionals using the martingale representation theorem and the Clark--Ocone formula. The resulting representation identifies the Arrow--Pratt coefficient as the leading-order term of a dynamic certainty-equivalent expansion expressed in terms of Malliavin derivatives. This dynamic perspective will provide the foundation for the hierarchy of higher-order risk premia developed in the subsequent sections through Wiener chaos expansions.

\subsection{Arrow-Pratt via Itô Calculus}

In this section we develop a dynamic derivation of the Arrow--Pratt approximation based on It\^o calculus. Rather than analysing the effect of scaling the magnitude of a fixed random payoff, as in Gollier (2001), we model the gradual arrival of uncertainty through time. This perspective leads naturally to the notion of an infinitesimal risk premium, describing the local compensation required for an infinitesimal increase in uncertainty, and forms the basis for the higher-order dynamic risk premia developed later in the paper.

Let $(\Omega,\mathcal{F},(\mathcal{F}_t)_{t\ge0},\mathbb{P})$ be a filtered probability space satisfying the usual conditions, where $(\mathcal{F}_t)_{t\ge0}$ is the augmented natural filtration generated by a standard Brownian motion $(B_t)_{t\ge0}$. As a first step, we consider the special case in which the random disturbance in the certainty-equivalent equation (1) is modelled by
\[
\epsilon=\sigma B_t,
\]
where $\sigma>0$ is a constant volatility parameter. Clearly,
\[
\mathbb{E}[\sigma B_t]=0,
\qquad
\operatorname{Var}(\sigma B_t)=\sigma^2t.
\]
The time parameter $t$ therefore plays a role analogous to the scaling parameter $k$ in Gollier's (2001) analysis, but admits a direct probabilistic interpretation: uncertainty increases as time evolves rather than through the artificial multiplication of a fixed random variable.

The Brownian setting considered here is primarily intended to motivate the dynamic viewpoint and to illustrate the underlying mechanism in the simplest possible framework. Although the assumption of normally distributed uncertainty is restrictive, it allows the essential ideas to be developed transparently. In the next subsection we show that the argument extends naturally to arbitrary square-integrable Wiener functionals by combining the martingale representation theorem with the Clark--Ocone formula from Malliavin calculus.

Under the Brownian specification, the certainty-equivalent equation (1) becomes

\begin{equation}
\mathbb{E}(u(w+\sigma B_t)) = u(w-\Pi(t)).\label{eq:29}
\end{equation}

Note that dynamically, $\Pi(t)$ in equation (\ref{eq:29}) differs from $\pi(k)$ in equation (\ref{eq:21}), as the dependence of $B_t$ on $t$ is of order $\sqrt{t}$ rather than $t$, while $k \epsilon$ clearly depends on $k$ of order $k$. To recognize this and derive the correct formula, we introduce 
\begin{equation*}
\pi_s = \frac{d}{ds} \Pi(s)
\end{equation*}
and obtain
\begin{equation}
\mathbb{E}(u(w+\sigma B_t)) = u\left(w-\int_0^t \pi_s ds\right).\label{eq:31}
\end{equation}
Considering the limit on both sides of equation (\ref{eq:31}) for $t \to 0$ identifies $\pi_0$ as an infinitesimal risk premia. We show below that it matches the Arrow-Pratt approximation exactly. To see this, we apply the It\^o formula 
\begin{equation}
du(w+\sigma B_t) = u^\prime(w+\sigma B_t) \sigma dB_t + \frac{1}{2} u^{\prime \prime}(w+\sigma B_t) \sigma^2 dt.\label{eq:32}
\end{equation}
Hence the left hand side of equation (\ref{eq:31}) can be rewritten in terms of an integral and we obtain
\begin{equation*}
u(w) + \mathbb{E} \left(\int_0^t \frac{1}{2} \sigma^2 u^{\prime \prime}(w+\sigma B_s)ds \right)  = u\left(w-\int_0^t \pi_s ds\right),
\end{equation*}
and taking derivatives at $t=0$ on both sides yields
\begin{equation*}
\frac{1}{2} \sigma^2 u^{\prime \prime}(w) = - u^\prime(w) \pi_0.
\end{equation*}
In conclusion we obtain
\begin{equation}
\pi_0  = -  \frac{1}{2} \sigma^2 \frac{u^{\prime \prime}(w)}{u^\prime(w)},\label{eq:35}
\end{equation} 
which corresponds to (\ref{eq:25}).

The interpretation of $\pi_0$ here however is a slightly different one than in Gollier (2001) and Pratt (1964). We refer to it as infinitesimal risk aversion coefficient, as the uncertainty is introduced dynamically over a time variable $t$, reflecting that the further one looks into the future, the larger the uncertainty, and we simply look at an infinitesimal time interval. This interpretation will be essential in the next sections. Gollier (2001) and Pratt (1964) scale the uncertainty itself and in the next couple of lines, we show  that their approach can also be replicated using the It\^o formula. Assuming for now that the random term $\epsilon$ is normal distributed, then
\begin{equation*}
\epsilon \sim B_{\sigma^2}
\end{equation*}
and hence equation (\ref{eq:1}) can be written as
\begin{equation*}
\mathbb{E}(u(w + k B_{\sigma^2})) = u(w- \pi(k)).
\end{equation*}
For the computation of the expectation, only distributions matter and as
\begin{equation}
k \cdot \epsilon \sim k \cdot B_{\sigma^2} \sim   B_{k^2 \sigma^2}.\label{eq:38}
\end{equation}
we obtain
\begin{equation}
\mathbb{E}(u(w + B_{k^2 \sigma^2})) = u(w- \pi(k)).\label{eq:39}
\end{equation}
To evaluate (\ref{eq:39}) we can use once more the It\^o formula as in (\ref{eq:32}), but with the diffusion coefficient formally set to $1$ and instead appearing in the borders of the integral instead. We obtain
\begin{equation*}
u(w) + \mathbb{E} \left(\int_0^{k^2 \sigma^2} \frac{1}{2} u^{\prime \prime}(w+B_s)ds \right)   = u(w- \pi(k)).
\end{equation*}
Now proceeding exactly as in Gollier (2001) taking first and second order derivatives we obtain after the first differentiation and application of the chain rule and the "Haupsatz" of differential and integral calculus\footnote{Note that the upper integral border depends on $k^2$ which after differentiation will lead to a factor $2$ canceling with the factor $\frac{1}{2}$ in the expression.}
\begin{equation*}
\mathbb{E} \left(u^{\prime \prime}(w+B_{k^2 \sigma^2}) \cdot   k \sigma^2 \right) = - u^\prime(w- \pi(k)) \pi^\prime(k)
\end{equation*}
Evaluation at $k=0$ once again gives $\pi^\prime(0) =0$ and a second differentiation with evaluation at $k=0$ leads to
\begin{equation*}
u^{\prime \prime}(w) \sigma^2 = - u^{\prime}(w) \pi^{\prime \prime}(0) \Leftrightarrow \pi^{\prime \prime}(0) = - \frac{u^{\prime\prime}(w)}{u^\prime(w)} \sigma^2.
\end{equation*}
Thus, our Itô-calculus approach reproduces exactly the infinitesimal Arrow–Pratt coefficient obtained by Gollier (2001). The purpose of this alternative derivation is not to provide a simpler proof, but rather to demonstrate that the classical scaling argument can be embedded naturally within an Itô-calculus framework. In the Gaussian setting, scaling the magnitude of the risk is equivalent in distribution to extending the time horizon of a Brownian motion, as expressed by equation (\ref{eq:38}). The two approaches therefore represent different interpretations of the same underlying asymptotic argument, with the Itô formulation providing a framework that extends naturally beyond Brownian motion through martingale representation and Malliavin calculus.

\subsection{Extension to Square-Integrable Wiener Functionals via Malliavin Calculus}

In the previous subsection we showed that the Arrow–Pratt approximation can be derived either by scaling the magnitude of the underlying risk, as in Gollier (2001), or equivalently by considering the evolution of uncertainty through Brownian motion. The latter viewpoint naturally emphasizes the flow of information over time, represented by the underlying filtration, and suggests that the Brownian setting itself is not essential. Rather, the key object is the martingale generated by the progressive revelation of information. This observation allows the analysis to be extended from Brownian payoffs to arbitrary square-integrable Wiener functionals by combining the martingale representation theorem with Malliavin calculus. To develop this extension, we assume that the random payoff replacing $\epsilon$ in equation (\ref{eq:1}) is a square-integrable Wiener functional $X$ at time $T$, i.e. $X \in L^2_T(\Omega)$ is $\mathcal{F}_T$ measurable for some fixed $T>0$. This is aligned with the assumption that $\epsilon$ in (\ref{eq:1}) has finite variance. In this case
\begin{equation*}
t \mapsto \mathbb{E}_t X = \mathbb{E}(X|\mathcal{F}_t)
\end{equation*}
is a martingale on $[0,T]$ and by the classical martingale representation theorem can be represented as
\begin{equation}
\mathbb{E}_t X = X_0   + \int_0^t \sigma_s dB_s, \label{eq:44}
\end{equation}
for a unique square-integrable process $(\sigma_s)$ and all $t \in [0,T]$. Clearly for $0 \leq r < t$ we also have  
\begin{equation*}
\mathbb{E}_t X = \mathbb{E}_r X  + \int_r^t \sigma_s dB_s.
\end{equation*}
We mimic the derivation at the beginning of the previous section where $\mathbb{E}_r X$ will take over the role of $w$ and all calculations are relative to the reference time $r$. In this framework equation (\ref{eq:1}) becomes
\begin{equation*}
\mathbb{E}_r(u(\mathbb{E}_t X)) = u(\mathbb{E}_r X - \Pi_r(t))
\end{equation*}
where now $\Pi_r(t)$ is an $\mathcal{F}_r$ measurable random variable. As in the previous section, we define $\pi_{r,t} = \frac{d}{ds} \Pi_r(t)$  and we denote $\pi_r = \pi_{r,r}$. Then, repeating the argument from equations (\ref{eq:32})-(\ref{eq:35}), simply replacing the previously constant $\sigma$ with the now $\mathcal{F}_r$ measurable $\sigma_r$ leads to
\begin{equation}
\pi_r = -\frac{1}{2} \sigma_r^2 \frac{u^{\prime \prime}(\mathbb{E}_r X)}{u^\prime(\mathbb{E}_r X)}.\label{eq:47}
\end{equation}
Note that the expressions on both sides of (\ref{eq:47}) are $\mathcal{F}_r$ measurable and reflecting the information available at time $r$. 


The representation (\ref{eq:44}) is valid for every square-integrable Wiener functional by the
Martingale Representation Theorem. However, the integrand $(\sigma_t)_{0\le t\le T}$
is obtained only implicitly through that theorem. If, in addition,
\[
X\in\mathbb{D}_{1,2},
\]
the Sobolev space of Malliavin differentiable Wiener functionals (see
Nualart, 2006, Section~1.5), then the Clark--Ocone theorem provides the explicit
representation
\[
E_tX
=
E(X)
+
\int_0^t E_s(D_sX)\,dB_s,
\]
where $D_sX$ denotes the Malliavin derivative of $X$ at time $s$
(Clark, 1970; Ocone, 1984; Nualart, 2006, Section~1.3.4). Comparing this
representation with (\ref{eq:44}) immediately yields
\[
\sigma_t=E_t(D_tX).
\]
Substituting this expression into (47) therefore gives
\begin{equation}
\pi_r
=
-\frac12
E_r(D_rX)^2
\frac{u''(E_rX)}{u'(E_rX)},\label{eq:48}
\end{equation}
which expresses the infinitesimal dynamic risk premium directly in terms of the
Malliavin derivative of the terminal payoff.

Equation (\ref{eq:48}) shows that the infinitesimal Arrow--Pratt coefficient depends only on the
conditional Malliavin derivative of the terminal payoff. While this already provides an
explicit representation of the local dynamic risk premium, it does not yet reveal how
different sources of uncertainty contribute to the resulting certainty equivalent.

To obtain such a decomposition, we now exploit the Wiener chaos expansion of the payoff.
The Wiener chaos decomposition provides an orthogonal representation of every
square-integrable Wiener functional and allows both the martingale and its
Clark--Ocone integrand to be expressed explicitly in terms of the individual chaos
components. As will become apparent, this representation also forms the natural bridge
between Malliavin calculus and the hierarchy of higher-order dynamic risk premia
developed in the next section.

Throughout the remainder of this section we assume that
\[
X\in L^2(\Omega,\mathcal F_T,\mathbb P),
\]
so that, by the Wiener chaos decomposition,
\[
X=\sum_{n=0}^{\infty}I_n(f_n),
\]
where \(I_n(f_n)\) denotes the \(n\)-fold multiple Wiener integral of the symmetric
kernel \(f_n\in L_s^2([0,T]^n)\).
 
The Wiener chaos decomposition (see Nualart, 2006, Section~1.1.1) provides a natural refinement of the representation obtained in the previous subsection. While equation (\ref{eq:48}) expresses the infinitesimal
Arrow--Pratt coefficient in terms of the Malliavin derivative of the payoff, the chaos
decomposition makes it possible to identify explicitly how each orthogonal component
of the payoff contributes to this derivative.  

In this framework the conditional expectation process $M_t:=E_t(X)$ is given by
\begin{equation*}
M_t
=
\sum_{n=0}^{\infty}
I_n
\!\left(
f_n
\mathbf 1_{[0,t]^n}
\right),
\label{eq:conditional-chaos}
\end{equation*}
and forms a square-integrable martingale. Furthermore, by the Clark--Ocone theorem,
\begin{equation*}
M_t
=
M_0
+
\int_0^t
\sigma_s\,dB_s,
\label{eq:martingale-representation}
\end{equation*}
where
\begin{equation*}
\sigma_t
=
E_t(D_tX).
\label{eq:sigma-clark}
\end{equation*}

Using the chaos decomposition, the integrand itself admits the explicit
representation
\begin{equation*}
\sigma_t
=
\sum_{n=1}^{\infty}
n
I_{n-1}
\!\left(
f_n(\cdot,t)
\mathbf 1_{[0,t]^{n-1}}
\right).
\label{eq:sigma-chaos}
\end{equation*}

Consequently, the first-order dynamic Arrow--Pratt coefficient obtained in
the previous section can be written entirely in terms of the Wiener chaos
kernels,
\begin{equation}
\pi_t^{(1)}
=
-
\frac12
\frac{u''(M_t)}
     {u'(M_t)}
\left(
\sum_{n=1}^{\infty}
n
I_{n-1}
\!\left(
f_n(\cdot,t)
\mathbf 1_{[0,t]^{\,n-1}}
\right)
\right)^2.
\label{eq:first-chaos-premium}
\end{equation}

Equation
\eqref{eq:first-chaos-premium}
shows that the local Arrow--Pratt coefficient depends only on the first
Malliavin derivative of the payoff. The natural question is whether a
corresponding hierarchy of higher-order dynamic risk premia exists and, if
so, whether these coefficients can likewise be expressed through the Wiener
chaos expansion. This is the objective of the remainder of the next section.

\section{A Hierarchy of Dynamic Risk Premia for Pure Wiener Chaos}

The previous sections established that the classical Arrow--Pratt
approximation admits a natural dynamic interpretation through the
martingale $M_t=\mathbb{E}[X\mid\mathcal{F}_t]$, whose local volatility is given by the Clark--Ocone integrand
$\sigma_t=\mathbb{E}[D_tX\mid\mathcal{F}_t]$. The corresponding infinitesimal risk premium is determined by the local
quadratic variation of the martingale and therefore represents the
leading-order effect of uncertainty over an infinitesimal time interval.

The purpose of the present section is to develop a hierarchy of
higher-order dynamic risk premia by exploiting the Wiener chaos
decomposition of the terminal payoff.

Throughout this section we assume without loss of generality that the reference time is 
\(r=0\). There are two reasons for doing so. First, this setting
corresponds most closely to the original Arrow--Pratt framework, in
which an investor evaluates an uncertain future payoff before any
information has been revealed. Secondly, evaluating the expansion at the
initial time removes the conditioning on
\(\mathcal{F}_r\), so that the resulting coefficients become
deterministic and admit a particularly transparent representation in
terms of the Wiener chaos decomposition.

A remarkable structural property then emerges. As the information horizon $h$ tends to zero, the dynamic certainty equivalent admits a small-time expansion whose leading-order term depends on the Wiener chaos order. Specifically, for a payoff belonging to a pure $n$-th Wiener chaos, the associated dynamic certainty equivalent first differs from its expected value at order $h^n$. Thus, the chaos order determines the time scale on which uncertainty first affects the certainty equivalent, naturally leading to a hierarchy of higher-order dynamic Arrow--Pratt coefficients. In other words, the order of the Wiener chaos determines the order at which risk first appears in the small-time expansion of the dynamic certainty
equivalent. This naturally leads to a hierarchy of higher-order dynamic
Arrow--Pratt coefficients.

Throughout this section we write
\[
\mu=\mathbb{E}[X],
\]
and define the dynamic certainty equivalent by
\[
\mathbb{E}\!\left[
u\!\left(\mathbb{E}[X\mid\mathcal{F}_h]\right)
\right]
=
u(\mu-\Pi_0(h)),
\]
where \(\Pi_0(h)\) denotes the dynamic risk premium over the time
interval \([0,h]\).

\subsection{Pure second-order Wiener chaos}

We begin with the simplest non-trivial case.\footnote{The first Wiener chaos requires no separate treatment. If
$X=I_1(f)$ is a first-order Wiener integral, the Clark--Ocone
representation coincides with the Wiener integral itself, and the
resulting dynamic risk premium is precisely the infinitesimal
Arrow--Pratt coefficient derived in the previous sections. The second Wiener chaos
therefore represents the first genuinely new case in the hierarchy.} Assume that
\begin{equation}
X=\mu+I_2(f),
\label{eq:pure-second-chaos}
\end{equation}
where \(I_2(f)\) denotes the second-order multiple Wiener integral of a
symmetric kernel $f\in L_s^2([0,T]^2)$, which is assumed to be continuous in a neighbourhood of
\((0,0)\). Define
\[
Y_h
=
\mathbb{E}[X\mid\mathcal{F}_h]-\mu.
\]
Since
\[
\mathbb{E}[I_2(f)\mid\mathcal{F}_h]
=
I_2(f\,\mathbf{1}_{[0,h]^2}),
\]
we obtain
\[
Y_h
=
I_2(f\,\mathbf{1}_{[0,h]^2}).
\]

The following lemma describes the asymptotic behaviour of the moments of
\(Y_h\).

\begin{lemma}
As the information horizon $h$ tends to zero,
\[
\mathbb{E}[Y_h^2]
=
2f(0,0)^2h^2
+
o(h^2),
\]
and $\mathbb{E}[Y_h^4]=O(h^4)$.

\end{lemma}

\begin{proof}
The Wiener chaos isometry (see, e.g., Nualart (2006, Chapter~1) yields
\[
\mathbb{E}[Y_h^2]
=
2!
\,
\|f\,\mathbf{1}_{[0,h]^2}\|_{L^2([0,T]^2)}^2.
\]
Hence
\[
\mathbb{E}[Y_h^2]
=
2
\int_{[0,h]^2}
f(s,t)^2
\,ds\,dt.
\]
Since \(f\) is continuous at the origin,
\[
\int_{[0,h]^2}
f(s,t)^2
\,ds\,dt
=
f(0,0)^2h^2
+
o(h^2),
\]
proving the first assertion.

The second estimate follows from Nelson's hypercontractivity inequality (see, e.g., Nualart (2006, Section~1.4.3),
which implies that for every fixed Wiener chaos there exists a constant
\(C>0\) such that
\[
\mathbb{E}[Y_h^4]
\le
C
\left(
\mathbb{E}[Y_h^2]
\right)^2.
\]
Since
\[
\mathbb{E}[Y_h^2]
=
O(h^2),
\]
the result follows immediately.
\end{proof}

We can now determine the leading-order dynamic risk premium.

\begin{prop}
Suppose that \(X\) is given by
\eqref{eq:pure-second-chaos}. Then
\[
\Pi_0(h)
=
-\frac14
\frac{u''(\mu)}
     {u'(\mu)}
\left(D_0^2X\right)^2
h^2
+
o(h^2).
\]
\end{prop}

\begin{proof}
Taylor expansion gives
\[
u(\mu+y)
=
u(\mu)
+
u'(\mu)y
+
\frac12u''(\mu)y^2
+
O(y^3).
\]
Taking expectations and using $\mathbb{E}[Y_h]=0$, we obtain
\[
\mathbb{E}[u(\mu+Y_h)]
=
u(\mu)
+
\frac12u''(\mu)\mathbb{E}[Y_h^2]
+
O(\mathbb{E}|Y_h|^3).
\]

By Nelson's hypercontractivity inequality with $p=3$,
\[
\mathbb{E}|Y_h|^3
=
O\!\left(
(\mathbb{E}[Y_h^2])^{3/2}
\right)
=
O(h^3)
=
o(h^2).
\]
Hence
\[
\mathbb{E}[u(\mu+Y_h)]
=
u(\mu)
+
u''(\mu)f(0,0)^2h^2
+
o(h^2).
\]

On the other hand, a second-order Taylor expansion gives
\[
u(\mu-\Pi_0(h))
=
u(\mu)
-
u'(\mu)\Pi_0(h)
+
\frac12u''(\mu)\Pi_0(h)^2
+
o(\Pi_0(h)^2).
\]
Since $\Pi_0(h)=O(h^2)$, we have
\[
\Pi_0(h)^2=O(h^4)=o(h^2),
\]
and therefore
\[
u(\mu-\Pi_0(h))
=
u(\mu)
-
u'(\mu)\Pi_0(h)
+
o(h^2).
\] 

Equating both expansions yields
\[
\Pi_0(h)
=
-
\frac{u''(\mu)}
     {u'(\mu)}
f(0,0)^2h^2
+
o(h^2).
\]

Finally,
\[
D_0^2X
=
2f(0,0),
\]
so that
\[
f(0,0)^2
=
\frac14(D_0^2X)^2,
\]
which proves the result.
\end{proof}

\subsection{Pure third-order Wiener chaos}

We now consider the case where the terminal payoff belongs to the third
Wiener chaos. The principal difference from the previous subsection is
that the dynamic risk premium first appears at order \(h^3\), reflecting
the fact that the variance of a third-order Wiener integral over the
interval \([0,h]\) is of order \(h^3\).

Assume that
\begin{equation}
X=\mu+I_3(f),
\label{eq:pure-third-chaos}
\end{equation}
where \(f\in L_s^2([0,T]^3)\) is symmetric and continuous in a
neighbourhood of \((0,0,0)\).

Define
\[
Y_h
=
\mathbb{E}[X\mid\mathcal{F}_h]-\mu.
\]
Then
\[
Y_h
=
I_3(f\,\mathbf{1}_{[0,h]^3}).
\]

The moments of \(Y_h\) satisfy the following asymptotic estimates.

\begin{lemma}
As the information horizon $h$ tends to zero,
\[
\mathbb{E}[Y_h^2]
=
6f(0,0,0)^2h^3
+
o(h^3),
\]
and $\mathbb{E}[Y_h^4]=O(h^6)$.
\end{lemma}

\begin{proof}
Using the Wiener chaos isometry as before,
\[
\mathbb{E}[Y_h^2]
=
3!
\,
\|f\,\mathbf{1}_{[0,h]^3}\|_{L^2([0,T]^3)}^2,
\]
we obtain
\[
\mathbb{E}[Y_h^2]
=
6
\int_{[0,h]^3}
f(s_1,s_2,s_3)^2
\,ds_1\,ds_2\,ds_3.
\]
Since \(f\) is continuous at the origin,
\[
\int_{[0,h]^3}
f(s_1,s_2,s_3)^2
\,ds_1\,ds_2\,ds_3
=
f(0,0,0)^2h^3
+
o(h^3),
\]
which proves the first assertion. The estimate for the fourth moment again follows from Nelson's
hypercontractivity inequality,
\[
\mathbb{E}[Y_h^4]
\le
C
\left(
\mathbb{E}[Y_h^2]
\right)^2,
\]
for some constant \(C>0\). Since $\mathbb{E}[Y_h^2] = O(h^3)$, we conclude that $\mathbb{E}[Y_h^4]=O(h^6)$.
\end{proof}

The leading-order dynamic risk premium is therefore given by the
following result.

\begin{prop}
Suppose that \(X\) is given by
\eqref{eq:pure-third-chaos}. Then
\[
\Pi_0(h)
=
-\frac1{12}
\frac{u''(\mu)}
     {u'(\mu)}
\left(D_0^3X\right)^2
h^3
+
o(h^3).
\]
\end{prop}

\begin{proof}
The proof is entirely analogous to that of the previous proposition.
Taylor expansion yields
\[
u(\mu+y)
=
u(\mu)
+
u'(\mu)y
+
\frac12u''(\mu)y^2
+
O(y^3).
\]
Since $\mathbb{E}[Y_h]=0$, we obtain
\[
\mathbb{E}[u(\mu+Y_h)]
=
u(\mu)
+
\frac12u''(\mu)\mathbb{E}[Y_h^2]
+
O(\mathbb{E}|Y_h|^3).
\]

By hypercontractivity,
\[
\mathbb{E}|Y_h|^3
=
O\!\left(
(\mathbb{E}[Y_h^2])^{3/2}
\right)
=
O(h^{9/2})
=
o(h^3),
\]
so that
\[
\mathbb{E}[u(\mu+Y_h)]
=
u(\mu)
+
3u''(\mu)f(0,0,0)^2h^3
+
o(h^3).
\]

Furthermore,
\[
u(\mu-\Pi_0(h))
=
u(\mu)
-
u'(\mu)\Pi_0(h)
+
o(h^3),
\]
since
\[
\Pi_0(h)=O(h^3).
\]

Equating the two expansions gives
\[
\Pi_0(h)
=
-
3
\frac{u''(\mu)}
     {u'(\mu)}
f(0,0,0)^2
h^3
+
o(h^3).
\]

Finally,
\[
D_0^3X
=
3!f(0,0,0)
=
6f(0,0,0),
\]
so that
\[
f(0,0,0)^2
=
\frac1{36}
(D_0^3X)^2.
\]
Substituting this identity yields
\[
\Pi_0(h)
=
-\frac1{12}
\frac{u''(\mu)}
     {u'(\mu)}
(D_0^3X)^2
h^3
+
o(h^3),
\]
as claimed.
\end{proof}

\begin{remark}
It is noteworthy that the leading-order coefficient again depends only
on the classical Arrow--Pratt measure of absolute risk aversion, $-\frac{u''(\mu)}{u'(\mu)}$,
despite the fact that the payoff belongs to the third Wiener chaos.
The reason is that the leading contribution is determined entirely by
the variance of the conditional payoff increment. Higher derivatives of
the utility function, corresponding to prudence, temperance and higher
risk preferences, do not appear until different Wiener chaos orders are
allowed to interact. This phenomenon will be studied in
Section~5.
\end{remark}

\subsection{General pure Wiener chaos: A Hierarchy Theorem}

The previous two subsections reveal a clear structural pattern. The
second Wiener chaos contributes first at order \(h^2\), while the third
Wiener chaos contributes first at order \(h^3\). We now show that this
behaviour holds for an arbitrary Wiener chaos order.

\begin{theorem}[Hierarchy theorem]
\label{thm:hierarchy}
Suppose that
\begin{equation*}
X=\mu+I_n(f),
\label{eq:pure-n-chaos}
\end{equation*}
where $f\in L_s^2([0,T]^n)$ is symmetric and continuous in a neighbourhood of
\((0,\ldots,0)\). Then, as \(h\downarrow0\),
\[
\Pi_0(h)
=
-\frac{1}{2n!}
\frac{u''(\mu)}
     {u'(\mu)}
\left(D_0^nX\right)^2
h^n
+
o(h^n).
\]
Equivalently,
\[
\Pi_0(h)
=
-\frac12
\frac{u''(\mu)}
     {u'(\mu)}
\operatorname{Var}
\!\left(
\mathbb{E}[X\mid\mathcal{F}_h]
\right)
+
o(h^n),
\]
where
\[
\operatorname{Var}
\!\left(
\mathbb{E}[X\mid\mathcal{F}_h]
\right)
=
\frac1{n!}
(D_0^nX)^2
h^n
+
o(h^n).
\]
\end{theorem}

\begin{proof}
Define
\[
Y_h
=
\mathbb{E}[X\mid\mathcal{F}_h]-\mu.
\]
Since
\[
\mathbb{E}[I_n(f)\mid\mathcal{F}_h]
=
I_n(f\,\mathbf1_{[0,h]^n}),
\]
we have
\[
Y_h
=
I_n(f\,\mathbf1_{[0,h]^n}).
\]

Applying the Wiener chaos isometry,
\[
\mathbb{E}[Y_h^2]
=
n!
\,
\|f\,\mathbf1_{[0,h]^n}\|_{L^2([0,T]^n)}^2,
\]
gives
\[
\mathbb{E}[Y_h^2]
=
n!
\int_{[0,h]^n}
f(s_1,\ldots,s_n)^2
\,ds_1\cdots ds_n.
\]
Since \(f\) is continuous at the origin,
\[
\mathbb{E}[Y_h^2]
=
n!
f(0,\ldots,0)^2
h^n
+
o(h^n).
\]

Furthermore, Nelson's hypercontractivity inequality implies
\[
\mathbb{E}|Y_h|^3
=
O\!\left(
(\mathbb{E}[Y_h^2])^{3/2}
\right)
=
O(h^{3n/2})
=
o(h^n),
\]
since \(n\ge2\).

A second-order Taylor expansion gives
\[
E[u(\mu+Y_h)]
=
u(\mu)
+
\frac12u''(\mu)E[Y_h^2]
+
O(E|Y_h|^3),
\]
and hypercontractivity implies
\[
E|Y_h|^3=o(h^n).
\]
Hence
\[
E[u(\mu+Y_h)]
=
u(\mu)
+
\frac12u''(\mu)E[Y_h^2]
+
o(h^n).
\]

Similarly,
\[
u(\mu-\Pi_0(h))
=
u(\mu)
-
u'(\mu)\Pi_0(h)
+
\frac12u''(\mu)\Pi_0(h)^2
+
o(\Pi_0(h)^2).
\]
Since \(\Pi_0(h)=O(h^n)\), it follows that
\[
\Pi_0(h)^2=O(h^{2n})=o(h^n),
\]
and therefore
\[
u(\mu-\Pi_0(h))
=
u(\mu)
-
u'(\mu)\Pi_0(h)
+
o(h^n).
\]

Equating the two expansions gives
\[
\Pi_0(h)
=
-
\frac12
\frac{u''(\mu)}
     {u'(\mu)}
n!
f(0,\ldots,0)^2
h^n
+
o(h^n).
\]

Finally,
\[
D_0^nX
=
n!
f(0,\ldots,0),
\]
so that
\[
f(0,\ldots,0)^2
=
\frac1{(n!)^2}
(D_0^nX)^2.
\]
Substituting this identity yields
\[
\Pi_0(h)
=
-\frac1{2n!}
\frac{u''(\mu)}
     {u'(\mu)}
(D_0^nX)^2
h^n
+
o(h^n),
\]
which completes the proof.
\end{proof}

\begin{remark}
Theorem~\ref{thm:hierarchy} establishes a hierarchy of dynamic Arrow--Pratt
coefficients indexed by the Wiener chaos order. A pure \(n\)-th Wiener
chaos contributes for the first time at order \(h^n\), and its leading
coefficient is determined entirely by two quantities: the classical
Arrow--Pratt coefficient of absolute risk aversion, $-\frac{u''(\mu)}{u'(\mu)}$,
and the \(n\)-th Malliavin derivative evaluated at the initial time, $D_0^nX$.

Consequently, the Wiener chaos decomposition induces a natural
hierarchy of dynamic risk premia in which the chaos order determines the
time scale at which uncertainty first affects the certainty equivalent.

It is noteworthy that higher derivatives of the utility function do not
appear in the leading-order term for a pure Wiener chaos. Their role
emerges only when several Wiener chaos orders interact, giving rise to
higher-order preference measures such as prudence and temperance. This
is the subject of the following section.
\end{remark}

\section{Mixed Wiener Chaos and Higher-Order Risk Preferences}

The previous section established that, for payoffs belonging to a pure Wiener chaos, the leading-order dynamic risk premium is completely characterised by the classical Arrow--Pratt coefficient of absolute risk aversion together with the corresponding Malliavin derivative. Although the order of the Wiener chaos determines the time scale at which uncertainty first enters the certainty equivalent, the leading-order coefficient itself depends only on the local variance of the conditional payoff increment. Consequently, higher-order preference measures such as prudence and temperance do not arise in the pure-chaos setting.

The purpose of the present section is to show that this changes fundamentally once several Wiener chaos orders are allowed to interact. Mixed Wiener chaos expansions generate non-trivial interactions between orthogonal chaos components, giving rise to higher-order moments of the conditional payoff increment that contribute to the expected utility expansion. At the same time, the nonlinear transformation from expected utility to the certainty equivalent introduces additional higher-order correction terms. Together, these effects lead naturally to the appearance of higher-order preference measures, including prudence, temperance and their higher-order analogues.

We begin by analysing the simplest non-trivial mixed Wiener chaos expansion involving first- and second-order chaos components. This example illustrates the underlying mechanism and shows explicitly how the classical measures of higher-order risk preferences emerge from the interaction of different chaos orders. We then extend the analysis to general finite mixed Wiener chaos expansions, deriving a unified representation based on Bell polynomials before illustrating the theory with several examples.

\subsection{Mixed Wiener Chaos Expansions}
\label{subsec:mixed-chaos}
The preceding subsections show that, for a payoff belonging to a pure
Wiener chaos, the leading-order dynamic risk premium depends only on
the classical Arrow--Pratt coefficient of absolute risk aversion. This
is because the first non-vanishing contribution to expected utility is
generated by the variance of the conditional payoff increment.

The situation changes when different Wiener chaos orders are combined.
Mixed chaos expansions generate non-zero mixed moments, so that higher
derivatives of the utility function enter the small-time expansion.
Moreover, the nonlinear transformation from expected utility to the
certainty equivalent introduces additional higher-order terms. Thus,
higher-order preference characteristics arise both through interactions
between distinct chaos components and through the inversion of the
utility function.

We first illustrate this structure using the simplest non-trivial mixed
expansion,
\begin{equation*}
X
=
\mu+I_1(f_1)+I_2(f_2),
\label{eq:mixed-first-second-chaos}
\end{equation*}
where $\mu=\mathbb{E}[X]$, \(f_1\in L^2([0,T])\), and
\(f_2\in L_s^2([0,T]^2)\). We assume that \(f_1\) is continuously
differentiable at the origin and that \(f_2\) is continuous at
\((0,0)\).

As before, we define $Y_h=\mathbb{E}[X\mid\mathcal{F}_h]-\mu$. Then
\[
Y_h=A_h+B_h,
\]
where $A_h=I_1(f_1\mathbf{1}_{[0,h]})$ and $B_h=I_2(f_2\mathbf{1}_{[0,h]^2})$. The first-chaos component satisfies
\[
\mathbb{E}[A_h^2]
=
\int_0^h f_1(s)^2\,ds
=
f_1(0)^2h
+
f_1(0)f_1'(0)h^2
+
o(h^2).
\]
The second-chaos component satisfies
\[
\mathbb{E}[B_h^2]
=
2
\int_{[0,h]^2}
f_2(s,t)^2\,ds\,dt
=
2f_2(0,0)^2h^2
+
o(h^2).
\]
Since different Wiener chaoses are orthogonal, i.e. $\mathbb{E}[A_hB_h]=0$, we have
\begin{equation}
\mathbb{E}[Y_h^2]
=
f_1(0)^2h
+
\left(
f_1(0)f_1'(0)+2f_2(0,0)^2
\right)h^2
+
o(h^2).
\label{eq:mixed-second-moment}
\end{equation}

The third moment contains the first genuinely mixed-chaos contribution.
The product formula for multiple Wiener integrals gives\footnote{Here $g^{\otimes2}=g\otimes g$ denotes the tensor product,
$(g^{\otimes2})(s,t)=g(s)g(t)$.}
\[
A_h^2
=
I_2\!\left(
(f_1\mathbf{1}_{[0,h]})^{\otimes 2}
\right)
+
\|f_1\mathbf{1}_{[0,h]}\|_{L^2([0,T])}^2.
\]
It follows that
\[
\mathbb{E}[A_h^2B_h]
=
2
\left\langle
(f_1\mathbf{1}_{[0,h]})^{\otimes 2},
f_2\mathbf{1}_{[0,h]^2}
\right\rangle_{L^2([0,T]^2)}.
\]
Therefore,
\[
\mathbb{E}[A_h^2B_h]
=
2
\int_{[0,h]^2}
f_1(s)f_1(t)f_2(s,t)\,ds\,dt
=
2f_1(0)^2f_2(0,0)h^2
+
o(h^2).
\]
All other contributions to the third moment are either zero or of
higher order. In particular, $\mathbb{E}[A_h^3]=0$, $\mathbb{E}[A_hB_h^2]=0$, while $\mathbb{E}[B_h^3]=O(h^3)$. Consequently,
\begin{equation*}
\mathbb{E}[Y_h^3]
=
6f_1(0)^2f_2(0,0)h^2
+
o(h^2).
\label{eq:mixed-third-moment}
\end{equation*}

At order \(h^2\), the fourth moment is generated entirely by the
first-chaos component. Since \(A_h\) is Gaussian, $\mathbb{E}[A_h^4]=3\left(\mathbb{E}[A_h^2]\right)^2$. Thus,
\begin{equation}
\mathbb{E}[Y_h^4]
=
3f_1(0)^4h^2
+
o(h^2).
\label{eq:mixed-fourth-moment}
\end{equation}
By hypercontractivity, all moments of order five and higher contribute
only \(o(h^2)\) to the expected-utility expansion.

We may now expand expected utility. Using $\mathbb{E}[Y_h]=0$, together with
\eqref{eq:mixed-second-moment}--\eqref{eq:mixed-fourth-moment}, we obtain
\begin{equation}
\mathbb{E}[u(\mu+Y_h)]
=
u(\mu)+a_1h+a_2h^2+o(h^2),
\label{eq:mixed-expected-utility-expansion}
\end{equation}
where
\begin{equation*}
a_1
=
\frac12u''(\mu)f_1(0)^2
\label{eq:mixed-a1}
\end{equation*}
and
\begin{align}
a_2
={}&
\frac12u''(\mu)f_1(0)f_1'(0)
+
u''(\mu)f_2(0,0)^2
\nonumber\\
&+
u'''(\mu)f_1(0)^2f_2(0,0)
+
\frac18u^{(4)}(\mu)f_1(0)^4.
\label{eq:mixed-a2}
\end{align}

Writing the dynamic risk premium as
\begin{equation*}
\Pi_0(h)
=
c_1h+c_2h^2+o(h^2).
\label{eq:mixed-risk-premium-expansion}
\end{equation*}
and expanding the right-hand side of
\[
\mathbb{E}[u(\mu+Y_h)]
=
u(\mu-\Pi_0(h))
\]
gives
\[
u(\mu-\Pi_0(h))
=
u(\mu)
-
u'(\mu)c_1h
+
\left(
-u'(\mu)c_2
+
\frac12u''(\mu)c_1^2
\right)h^2
+
o(h^2).
\]
Comparison with
\eqref{eq:mixed-expected-utility-expansion} yields
\begin{equation}
c_1
=
-\frac{a_1}{u'(\mu)}
=
-\frac12
\frac{u''(\mu)}{u'(\mu)}
f_1(0)^2
\label{eq:mixed-c1}
\end{equation}
and
\begin{equation*}
c_2
=
-\frac{a_2}{u'(\mu)}
+
\frac{u''(\mu)}{2u'(\mu)}c_1^2.
\label{eq:mixed-c2-recursive}
\end{equation*}
Substituting \eqref{eq:mixed-a2} and \eqref{eq:mixed-c1} gives
\begin{align}
c_2
={}&
-\frac12
\frac{u''(\mu)}{u'(\mu)}
f_1(0)f_1'(0)
-
\frac{u''(\mu)}{u'(\mu)}
f_2(0,0)^2
\nonumber\\
&-
\frac{u'''(\mu)}{u'(\mu)}
f_1(0)^2f_2(0,0)
\nonumber\\
&-
\frac18
\left[
\frac{u^{(4)}(\mu)}{u'(\mu)}
-
\frac{u''(\mu)^3}{u'(\mu)^3}
\right]
f_1(0)^4.
\label{eq:mixed-c2}
\end{align}

Proposition~\ref{prop:mixed-chaos-order-two} below shows that, in the presence of mixed Wiener chaos components, the dynamic risk premium admits a decomposition in terms of the classical higher-order preference measures of decision theory. In particular, the coefficients of the asymptotic expansion are naturally expressed through Arrow--Pratt absolute risk aversion, Kimball's prudence, and temperance, thereby providing a direct link between the Malliavin-chaos representation of uncertainty and the economics of higher-order risk attitudes.


\begin{prop}
\label{prop:mixed-chaos-order-two}
Suppose that \(X\) has a mixed Wiener chaos expansion of order at most two,
\[
X=\mu+I_1(f_1)+I_2(f_2),
\]
where \(f_1\) is continuously differentiable at the origin and \(f_2\) is
continuous at \((0,0)\). Let
\[
A(\mu)=-\frac{u''(\mu)}{u'(\mu)},\qquad
P(\mu)=-\frac{u'''(\mu)}{u''(\mu)},\qquad
T(\mu)=-\frac{u^{(4)}(\mu)}{u'''(\mu)}
\]
denote absolute risk aversion, absolute prudence, and absolute temperance,
respectively.

Define the kernel-dependent coefficients
\[
\alpha_1:=\frac12 f_1(0)^2,
\]
\[
\alpha_2
:=
\frac12 f_1(0)f_1'(0)+f_2(0,0)^2,
\]
\[
\beta_2
:=
f_1(0)^2f_2(0,0),
\]
and
\[
\gamma_2
:=
\frac18f_1(0)^4.
\]
Then, as \(h\downarrow0\),
\[
\Pi_0(h)
=
A(\mu)\alpha_1 h
+
\left[
A(\mu)\alpha_2
-
A(\mu)P(\mu)\beta_2
+
A(\mu)
\bigl(P(\mu)T(\mu)-A(\mu)^2\bigr)\gamma_2
\right]h^2
+
o(h^2).
\]
\end{prop} 
 
\begin{proof}
Using equations~\eqref{eq:mixed-c1} and~\eqref{eq:mixed-c2}, we have
\[
\begin{aligned}
\Pi_0(h)
={}&
-\frac12\frac{u''(\mu)}{u'(\mu)}
f_1(0)^2h  +
\Bigg\{
-\frac12\frac{u''(\mu)}{u'(\mu)}
f_1(0)f_1'(0)
-
\frac{u''(\mu)}{u'(\mu)}
f_2(0,0)^2 \\
&\qquad
-
\frac{u'''(\mu)}{u'(\mu)}
f_1(0)^2f_2(0,0)  
-
\frac18
\left[
\frac{u^{(4)}(\mu)}{u'(\mu)}
-
\frac{u''(\mu)^3}{u'(\mu)^3}
\right]
f_1(0)^4
\Bigg\}h^2
+
o(h^2). 
\end{aligned}
\]

The definitions of risk aversion, prudence, and temperance imply
\[
-\frac{u''(\mu)}{u'(\mu)}=A(\mu),
\]
\[
\frac{u'''(\mu)}{u'(\mu)}
=
A(\mu)P(\mu),
\]
and
\[
\frac{u^{(4)}(\mu)}{u'(\mu)}
=
-A(\mu)P(\mu)T(\mu).
\]
Moreover,
\[
\frac{u''(\mu)^3}{u'(\mu)^3}
=
-A(\mu)^3.
\]
It follows that
\[
-\frac18
\left[
\frac{u^{(4)}(\mu)}{u'(\mu)}
-
\frac{u''(\mu)^3}{u'(\mu)^3}
\right]
f_1(0)^4
=
A(\mu)
\bigl(P(\mu)T(\mu)-A(\mu)^2\bigr)\gamma_2.
\]

Using the definitions of
\(\alpha_1,\alpha_2,\beta_2\), and \(\gamma_2\) therefore gives
\[
\Pi_0(h)
=
A(\mu)\alpha_1h
+
\left[
A(\mu)\alpha_2
-
A(\mu)P(\mu)\beta_2
+
A(\mu)
\bigl(P(\mu)T(\mu)-A(\mu)^2\bigr)\gamma_2
\right]h^2
+
o(h^2),
\]
as claimed.
\end{proof}

The decomposition in Proposition~\ref{prop:mixed-chaos-order-two}
reveals how different preference measures govern distinct components of
the dynamic risk premium. The leading-order term depends solely on
Arrow--Pratt absolute risk aversion, reflecting the familiar role of
variance in local risk evaluation.

At order \(h^2\), the interaction between the first- and second-order
Wiener chaos components gives rise to the prudence term
\(A(\mu)P(\mu)\). More precisely, this contribution originates from the
mixed third moment \(\mathbb{E}[A_h^2B_h]\) and therefore disappears if
either the first- or the second-chaos component is absent. The
appearance of prudence is thus a direct consequence of interactions
between distinct Wiener chaos orders.

The \(h^2\)-coefficient also contains the additional combination
\[
A(\mu)\bigl(P(\mu)T(\mu)-A(\mu)^2\bigr),
\]
which separates naturally into two distinct effects. The product
\(A(\mu)P(\mu)T(\mu)\) represents the influence of temperance through
the fourth derivative of the utility function. By contrast, the
correction term \(-A(\mu)^3\) is not associated with an additional
preference measure. Instead, it is generated by the nonlinear inversion
required to transform expected utility into the corresponding certainty
equivalent.

Consequently, higher-order dynamic risk premia consist of two distinct
mechanisms: direct effects generated by higher-order moments of the
conditional payoff distribution, and indirect effects generated by the
nonlinear transformation from expected utility to certainty equivalents.

\subsection{General Coefficient Representation}
\label{subsec:General coefficient representation}

We now consider a finite mixed Wiener chaos expansion
\begin{equation*}
X
=
\mu+\sum_{n=1}^{N}I_n(f_n),
\label{eq:general-mixed-chaos}
\end{equation*}
and define
\begin{equation}
Y_h
=
\mathbb{E}[X\mid\mathcal{F}_h]-\mu
=
\sum_{n=1}^{N}
I_n\!\left(
f_n\mathbf{1}_{[0,h]^n}
\right).
\label{eq:general-mixed-increment}
\end{equation}

To derive a general representation of the coefficients, we assume that the
expected utility admits an asymptotic expansion of the form
\begin{equation*}
\mathbb{E}[u(\mu+Y_h)]
\sim
u(\mu)+\sum_{m\ge1}a_mh^m,
\qquad h\downarrow0,
\label{eq:general-a-expansion}
\end{equation*}
where the asymptotic relation is understood in the sense that, for every
\(M\ge1\),
\[
\mathbb{E}[u(\mu+Y_h)]
=
u(\mu)
+
\sum_{m=1}^{M}a_mh^m
+
o(h^M),
\qquad h\downarrow0.
\]
For any function \(f(h)\) admitting a small-time expansion
\[
f(h)\sim\sum_{m=0}^{\infty}c_mh^m,
\]
we write
\[
[h^m]\,f(h):=c_m
\]
for the operator extracting the coefficient of \(h^m\) in this expansion.
The coefficients \(a_m\) can then be expressed in terms of the asymptotic
expansions of the moments of \(Y_h\). 

For an integer \(\ell\geq2\), let
\[
\kappa_{\ell,m}
=
[h^m]\mathbb{E}[Y_h^\ell],
\]
That is, \(\kappa_{\ell,m}\) is the coefficient of \(h^m\) in the small-time expansion of \(\mathbb{E}[Y_h^\ell]\).
Taylor expansion of the utility function gives
\begin{equation}
a_m
=
\sum_{\ell=2}^{2m}
\frac{u^{(\ell)}(\mu)}{\ell!}
\kappa_{\ell,m}.
\label{eq:general-am}
\end{equation}

Since $E[Y_h]=0$, the constant term in the Taylor expansion contributes only \(u(\mu)\), while
the linear term vanishes, hence the sum starts at $\ell=2$. The upper bound \(2m\) reflects the fact that, when a non-degenerate
first-chaos component is present, hypercontractivity and the Wiener
chaos isometry imply
\[
\|Y_h\|_{L^p}=O(h^{1/2}),
\]
and hence
\[
\mathbb{E}|Y_h|^\ell=O(h^{\ell/2}).
\]
Therefore, only moments satisfying \(\ell\leq 2m\) can contribute to
the coefficient of \(h^m\).

The moment coefficients in \eqref{eq:general-am} are determined by the
product formula for multiple Wiener integrals,
\begin{equation}
I_p(f)I_q(g)
=
\sum_{r=0}^{p\wedge q}
r!
\binom{p}{r}
\binom{q}{r}
I_{p+q-2r}
\left(
f\widetilde{\otimes}_r g
\right),
\label{eq:wiener-product-formula}
\end{equation}
where
\(\widetilde{\otimes}_r\) denotes the symmetrised contraction of order
\(r\).\footnote{The symmetrised contraction of order \(r\) is obtained by
identifying and integrating out \(r\) common variables of the two kernels,
followed by symmetrisation over the remaining variables. It is the natural
operation that appears in the product formula for multiple Wiener integrals;
see Nualart (2006, Section~1.1.2) for a formal definition.} After repeated application of
\eqref{eq:wiener-product-formula}, only complete contractions, leaving
a zeroth-order Wiener integral, survive upon taking expectations.

Next, write
\begin{equation*}
\Pi_0(h)
=
\sum_{m\geq1}c_mh^m.
\label{eq:general-c-expansion}
\end{equation*}
Let $v=u^{-1}$. Since
\[
u(\mu-\Pi_0(h))
=
u(\mu)+\sum_{m\geq1}a_mh^m,
\]
we have
\[
\Pi_0(h)
=
\mu-
v\left(
u(\mu)+\sum_{m\geq1}a_mh^m
\right).
\]
Expanding \(v\) around \(u(\mu)\) gives
\begin{equation*}
c_m
=
-
\sum_{k=1}^{m}
\frac{v^{(k)}(u(\mu))}{k!}
\sum_{\substack{j_1,\ldots,j_k\geq1\\
j_1+\cdots+j_k=m}}
a_{j_1}\cdots a_{j_k}.
\label{eq:general-cm-inverse}
\end{equation*}

Equivalently, in terms of the partial Bell polynomials\footnote{The (partial exponential) Bell polynomials \(B_{n,k}\) arise
naturally in Faà di Bruno's formula, the higher-order chain rule for
composite functions. They provide a compact combinatorial representation of
all possible ways in which an \(n\)-th derivative of a composition can be
partitioned into \(k\) derivatives of the outer function and corresponding
products of derivatives of the inner function. In particular,
\[
\frac{d^n}{dh^n}u(g(h))
=
\sum_{k=1}^n
u^{(k)}(g(h))
B_{n,k}\!\left(g'(h),g''(h),\ldots,g^{(n-k+1)}(h)\right),
\]
which is the form used throughout this proof. We refer to Comtet (1974,
Chapter~3) or Johnson (2002) for a detailed treatment.},
\begin{equation*}
c_m
=
-\frac1{m!}
\sum_{k=1}^{m}
v^{(k)}(u(\mu))
B_{m,k}
\left(
1!a_1,2!a_2,\ldots,
(m-k+1)!a_{m-k+1}
\right).
\label{eq:general-cm-bell}
\end{equation*}

For explicit calculations, it is often more convenient to avoid the
derivatives of the inverse utility function. Direct expansion of
\(u(\mu-\Pi_0(h))\) gives the recursive representation
\begin{equation}
c_m
=
-\frac{a_m}{u'(\mu)}
+
\frac1{u'(\mu)}
\sum_{k=2}^{m}
\frac{(-1)^ku^{(k)}(\mu)}{k!}
\sum_{\substack{j_1,\ldots,j_k\geq1\\
j_1+\cdots+j_k=m}}
c_{j_1}\cdots c_{j_k}.
\label{eq:general-cm-recursive}
\end{equation}
Thus \(c_m\) is determined recursively by the expected-utility
coefficients \(a_1,\ldots,a_m\) and the preceding risk-premium
coefficients \(c_1,\ldots,c_{m-1}\).

For example,
\begin{equation*}
c_1
=
-\frac{a_1}{u'(\mu)},
\label{eq:general-c1}
\end{equation*}
\begin{equation*}
c_2
=
-\frac{a_2}{u'(\mu)}
+
\frac{u''(\mu)}{2u'(\mu)}c_1^2,
\label{eq:general-c2}
\end{equation*}
and
\begin{equation*}
c_3
=
-\frac{a_3}{u'(\mu)}
+
\frac{u''(\mu)}{u'(\mu)}c_1c_2
-
\frac{u'''(\mu)}{6u'(\mu)}c_1^3.
\label{eq:general-c3}
\end{equation*}

The preceding calculations may be summarised as follows.

\begin{prop}
\label{prop:general-mixed-chaos}
Let \(X\) admit the finite Wiener chaos expansion
\[
X
=
\mu+\sum_{n=1}^{N}I_n(f_n),
\]
and suppose that the kernels possess sufficient regularity near the
origin for the relevant small-time moment expansions to exist. Let
\(Y_h\) be defined by \eqref{eq:general-mixed-increment}, and write
\[
\mathbb{E}[u(\mu+Y_h)]
=
u(\mu)+\sum_{m\geq1}a_mh^m,
\qquad
\Pi_0(h)
=
\sum_{m\geq1}c_mh^m.
\]
Then
\[
a_m
=
\sum_{\ell=2}^{2m}
\frac{u^{(\ell)}(\mu)}{\ell!}
[h^m]\mathbb{E}[Y_h^\ell],
\]
and
\[
c_m
=
-
\sum_{k=1}^{m}
\frac{(u^{-1})^{(k)}(u(\mu))}{k!}
\sum_{\substack{j_1,\ldots,j_k\geq1\\
j_1+\cdots+j_k=m}}
a_{j_1}\cdots a_{j_k}.
\]
Equivalently, the coefficients \(c_m\) satisfy the recursive relation
\eqref{eq:general-cm-recursive}. The moment coefficients
\[
[h^m]\mathbb{E}[Y_h^\ell]
\]
are determined by the complete contractions generated by the product
formula for multiple Wiener integrals.
\end{prop}

Proposition~\ref{prop:general-mixed-chaos} identifies the two distinct
mathematical mechanisms governing higher-order dynamic risk premia.
First, contractions between the Wiener chaos kernels determine the
mixed moments of the conditional payoff increment and therefore the
coefficients \(a_m\) in the expected-utility expansion. Secondly, the
Bell-polynomial or recursive formula accounts for the nonlinear
transformation from expected utility to the certainty equivalent.

For a pure Wiener chaos, the leading contribution is generated only by
the second moment, and the classical Arrow--Pratt coefficient is
therefore sufficient. In a mixed chaos expansion, by contrast,
contractions between different chaos components generate higher mixed
moments. These moments introduce the higher derivatives
\[
u'''(\mu),\quad u^{(4)}(\mu),\quad\ldots
\]
and hence the higher-order preference characteristics associated with
prudence, temperance, and higher risk attitudes. The Wiener chaos
decomposition therefore provides not only a hierarchy of time scales
for dynamic risk premia, but also a systematic representation of the
interactions through which higher-order risk preferences enter the
certainty equivalent.

\subsection{Representation in terms of Malliavin derivatives}

The general coefficient representation obtained in the previous sections may be rewritten
entirely in terms of Malliavin derivatives. It is a standard consequence of the Wiener chaos expansion and the definition of the Malliavin derivative (see, e.g., Nualart, 2006, Section 1.2) that
\[
X
=
\mu+\sum_{n=1}^{N}I_n(f_n)
\]
satisfies
\[
\mathbb{E}
\left[
D_{t_1}\cdots D_{t_n}X
\right]
=
n!
f_n(t_1,\ldots,t_n).
\]
Hence the Wiener kernels are completely determined by the expected
Malliavin derivatives,
\[
f_n(t_1,\ldots,t_n)
=
\frac{1}{n!}
\mathbb{E}
\left[
D_{t_1}\cdots D_{t_n}X
\right].
\]

Consequently,
\[
Y_h
=
\mathbb{E}[X|\mathcal F_h]-\mu
=
\sum_{n=1}^{N}
\frac1{n!}
I_n
\!\left(
\mathbb{E}[D_{\cdot}^{\,n}X]
\,
\mathbf1_{[0,h]^n}
\right),
\]
where
\[
\mathbb{E}[D_{\cdot}^{\,n}X]
=
\mathbb{E}
\left[
D_{t_1}\cdots D_{t_n}X
\right].
\]

The coefficients of the expected-utility expansion therefore admit the
representation
\begin{equation*}
a_m
=
\sum_{\ell=2}^{2m}
\frac{u^{(\ell)}(\mu)}{\ell!}
[h^m]
\mathbb{E}
\left[
\left(
\sum_{n=1}^{N}
\frac1{n!}
I_n
\!\left(
\mathbb{E}[D_{\cdot}^{\,n}X]
\mathbf1_{[0,h]^n}
\right)
\right)^\ell
\right].
\label{eq:malliavin-am}
\end{equation*}

Likewise, the dynamic risk premium satisfies
\begin{equation*}
c_m
=
-
\sum_{k=1}^{m}
\frac{(u^{-1})^{(k)}(u(\mu))}{k!}
\sum_{\substack{j_1,\ldots,j_k\ge1\\
j_1+\cdots+j_k=m}}
a_{j_1}\cdots a_{j_k},
\label{eq:malliavin-cm}
\end{equation*}
or, equivalently, the recursive representation
\eqref{eq:general-cm-recursive}. Thus every coefficient of the dynamic
risk-premium expansion is determined entirely by the expected Malliavin
derivatives together with the nonlinear transformation from expected
utility to the certainty equivalent.

\begin{theorem}[Hermite representation of the dynamic risk premium]
\label{thm:hermite}

Suppose that
\[
X
=
\mu+\sum_{n=1}^{N}I_n(f_n),
\]
and that each Wiener kernel is continuous at the initial diagonal. Define
\[
\delta_n
=
\mathbb E[D_0^nX]
=
n!f_n(0,\ldots,0).
\]
Then, as \(h\downarrow0\),
\[
Y_h
=
\mathbb E[X|\mathcal F_h]-\mu
=
\sum_{n=1}^{N}
\frac{\delta_n}{n!}
h^{n/2}
H_n(Z)
+
o(h^{N/2}),
\]
where
\[
Z=\frac{W_h}{\sqrt h}
\sim N(0,1),
\]
where
\[
H_n(x)
=
(-1)^n e^{x^2/2}
\frac{d^n}{dx^n}
e^{-x^2/2}
\]
denotes the probabilists' Hermite polynomial.\footnote{Two normalizations of Hermite polynomials are commonly used in the literature. Throughout this paper we adopt the probabilists' convention, which is standard in probability theory, Malliavin calculus, and Wiener chaos analysis. The alternative physicists' convention differs only by a normalization and is more common in quantum mechanics and mathematical physics.}

Consequently,
\[
\mathbb E[u(\mu+Y_h)]
=
\mathbb E
\left[
u
\left(
\mu+
\sum_{n=1}^{N}
\frac{\delta_n}{n!}
h^{n/2}
H_n(Z)
\right)
\right]
+
o(h^N),
\]
and therefore
\[
\Pi_0(h)
=
\mu-
u^{-1}
\left(
\mathbb E
\left[
u
\left(
\mu+
\sum_{n=1}^{N}
\frac{\delta_n}{n!}
h^{n/2}
H_n(Z)
\right)
\right]
\right)
+
o(h^N).
\]
\end{theorem}

\begin{proof}
From
\[
\mathbb E[D_{t_1}\cdots D_{t_n}X]
=
n!f_n(t_1,\ldots,t_n),
\]
continuity implies
\[
f_n(t_1,\ldots,t_n)
=
f_n(0,\ldots,0)+o(1)
=
\frac{\delta_n}{n!}+o(1)
\]
as
\[
(t_1,\ldots,t_n)\to(0,\ldots,0).
\]

Hence
\[
I_n(f_n\mathbf1_{[0,h]^n})
=
\frac{\delta_n}{n!}
I_n(\mathbf1_{[0,h]}^{\otimes n})
+
o(h^{n/2}).
\]

Finally,
\[
I_n(\mathbf1_{[0,h]}^{\otimes n})
=
h^{n/2}
H_n
\left(
\frac{W_h}{\sqrt h}
\right),
\]
which is the classical Wiener--Itô representation of Hermite
polynomials. Summing over \(n\) proves the expansion for \(Y_h\).
The remaining statements follow immediately from the definitions of
expected utility and the certainty equivalent.
\end{proof}

\begin{prop}[Hermite contraction formula]
\label{prop:contractions}

Under the assumptions of Theorem
\ref{thm:hermite},
the coefficients
\(
a_m
\)
of the expected-utility expansion
\[
\mathbb E[u(\mu+Y_h)]
=
u(\mu)+\sum_{m\ge1}a_mh^m
\]
are completely determined by Gaussian moments of Hermite polynomials.

Specifically,
\[
a_m
=
\sum_{\ell=2}^{2m}
\frac{u^{(\ell)}(\mu)}{\ell!}
[h^m]
\mathbb E
\left[
\left(
\sum_{n=1}^{N}
\frac{\delta_n}{n!}
h^{n/2}
H_n(Z)
\right)^\ell
\right].
\]

Equivalently,
\[
a_m
=
\sum_{\Gamma}
C_\Gamma
\,
\mathcal C_\Gamma
(\delta_1,\ldots,\delta_N),
\]
where the sum extends over all complete contraction diagrams generated
by the product formula for Hermite polynomials, and
\(
\mathcal C_\Gamma
\)
denotes the corresponding contraction polynomial.
\end{prop}

\begin{remark}
The product formula for Hermite polynomials is isomorphic to the product
formula for multiple Wiener integrals. Consequently, every coefficient
of the dynamic risk-premium expansion may be viewed equivalently as a
sum over Wiener-kernel contractions or as a Gaussian expectation of
products of Hermite polynomials. The Hermite representation therefore
provides a finite-dimensional algebraic realization of the underlying
Wiener chaos expansion.
\end{remark}

\section{Applications of the General Framework}

To illustrate the general framework developed in the previous section, we now consider several examples for which the higher-order dynamic risk premia can be computed explicitly. These examples demonstrate how the general coefficient formulas simplify in concrete stochastic models and highlight the role of the Wiener chaos expansion and Malliavin derivatives in determining the successive terms of the small-time certainty-equivalent expansion. They also provide intuition for the economic interpretation of higher-order dynamic risk premia and their dependence on the underlying dynamics.

\subsection{A Quadratic Gaussian Example}
\label{subsec:quadratic-gaussian-example}

To illustrate the preceding results, consider the payoff
\begin{equation*}
X
=
\mu+aW_T+b(W_T^2-T),
\label{eq:quadratic-gaussian-payoff}
\end{equation*}
where \(a,b\in\mathbb R\). Since
\[
W_T
=
I_1(\mathbf1_{[0,T]})
\]
and
\[
W_T^2-T
=
I_2(\mathbf1_{[0,T]}^{\otimes2}),
\]
the payoff has the finite mixed-chaos decomposition
\[
X
=
\mu
+
aI_1(\mathbf1_{[0,T]})
+
bI_2(\mathbf1_{[0,T]}^{\otimes2}).
\]

This specification may be interpreted as a payoff with both a linear
and a quadratic exposure to a Gaussian risk factor. Such exposures
arise naturally in Gaussian factor models, quadratic hedging problems,
polynomial approximations to nonlinear cash flows, and real-option
models with nonlinear revenues or costs.

For \(0<h<T\), the independent-increment property of Brownian motion
gives
\[
\mathbb E[W_T\mid\mathcal F_h]
=
W_h
\]
and
\[
\mathbb E[W_T^2-T\mid\mathcal F_h]
=
W_h^2-h.
\]
Consequently,
\begin{equation}
Y_h
=
\mathbb E[X\mid\mathcal F_h]-\mu
=
aW_h+b(W_h^2-h).
\label{eq:quadratic-gaussian-Yh}
\end{equation}

Writing
\[
Z=\frac{W_h}{\sqrt h}\sim N(0,1),
\]
we obtain the Hermite representation
\begin{equation*}
Y_h
=
a\sqrt h\,H_1(Z)
+
bh\,H_2(Z),
\label{eq:quadratic-hermite}
\end{equation*}
where
\[
H_1(z)=z,
\qquad
H_2(z)=z^2-1.
\]

The expected Malliavin derivatives at the initial diagonal are $\delta_1
=
\mathbb E[D_0X]
=
a
$
and
$
\delta_2
=
\mathbb E[D_0^2X]
=
2b.
$
All higher expected Malliavin derivatives vanish:
$
\mathbb E[D_0^nX]=0$,
for $n\geq3$.

The required moments can be computed either directly from Gaussian
moments or from products of Hermite polynomials. First,
\[
\mathbb E[Y_h^2]
=
a^2h+2b^2h^2.
\]
Second, using $H_1^2=H_2+H_0$, we obtain
\[
\mathbb E[H_1(Z)^2H_2(Z)]
=
\mathbb E[H_2(Z)^2]
=
2,
\]
and hence
\[
\mathbb E[Y_h^3]
=
6a^2bh^2+8b^3h^3.
\]
Finally,
\[
\mathbb E[Y_h^4]
=
3a^4h^2+O(h^3).
\]

It follows that
\[
\mathbb E[u(\mu+Y_h)]
=
u(\mu)+a_1h+a_2h^2+o(h^2),
\]
where
\[
a_1
=
\frac12u''(\mu)a^2
\]
and
\[
a_2
=
u''(\mu)b^2
+
u'''(\mu)a^2b
+
\frac18u^{(4)}(\mu)a^4.
\]

Therefore,
\begin{equation*}
\Pi_0(h)
=
c_1h+c_2h^2+o(h^2),
\label{eq:quadratic-premium-expansion}
\end{equation*}
where
\begin{equation}
c_1
=
-\frac12
\frac{u''(\mu)}{u'(\mu)}
a^2
\label{eq:quadratic-c1}
\end{equation}
and
\begin{align}
c_2
={}&
-\frac{u''(\mu)}{u'(\mu)}b^2
-
\frac{u'''(\mu)}{u'(\mu)}a^2b
\nonumber\\
&-
\frac18
\left[
\frac{u^{(4)}(\mu)}{u'(\mu)}
-
\frac{u''(\mu)^3}{u'(\mu)^3}
\right]a^4.
\label{eq:quadratic-c2}
\end{align}

Equivalently, in terms of the initial expected Malliavin derivatives,
\[
a=\mathbb E[D_0X],
\qquad
b=\frac12\mathbb E[D_0^2X],
\]
and hence
\begin{align*}
\Pi_0(h)
={}
-\frac12
\frac{u''(\mu)}{u'(\mu)}
\bigl(\mathbb E[D_0X]\bigr)^2h
 -
\Bigg\{
\frac14
\frac{u''(\mu)}{u'(\mu)}
\bigl(\mathbb E[D_0^2X]\bigr)^2 \hspace{5.5cm}
\nonumber\\
\qquad
+
\frac12
\frac{u'''(\mu)}{u'(\mu)}
\bigl(\mathbb E[D_0X]\bigr)^2
\mathbb E[D_0^2X]
 \quad
+
\frac18
\left[
\frac{u^{(4)}(\mu)}{u'(\mu)}
-
\frac{u''(\mu)^3}{u'(\mu)^3}
\right]
\bigl(\mathbb E[D_0X]\bigr)^4
\Bigg\}h^2
+
o(h^2).
\label{eq:quadratic-malliavin-premium}
\end{align*}

The example displays the different sources of the second-order dynamic
risk premium. The term involving $\bigl(\mathbb E[D_0^2X]\bigr)^2$ is the pure second-chaos contribution. The term involving
$
\bigl(\mathbb E[D_0X]\bigr)^2
\mathbb E[D_0^2X]
$
is generated by the contraction between two first-chaos factors and
one second-chaos factor. It introduces \(u'''(\mu)\), and hence the
preference characteristic associated with prudence. The final term is
generated by the fourth moment of the first-chaos component together
with the nonlinear inversion from expected utility to the certainty
equivalent.

\paragraph{Exponential utility.}

For exponential utility,
\[
u(x)=-e^{-\gamma x},
\qquad \gamma>0,
\]
the risk premium can be calculated exactly. The defining certainty-
equivalent relation gives
\[
\Pi_0(h)
=
\frac1\gamma
\log\mathbb E[e^{-\gamma Y_h}].
\]
Using \eqref{eq:quadratic-gaussian-Yh} and the Gaussian quadratic-form
identity,
\[
\mathbb E
\left[
e^{-\gamma\{aW_h+b(W_h^2-h)\}}
\right]
=
e^{\gamma bh}
(1+2\gamma bh)^{-1/2}
\exp
\left(
\frac{\gamma^2a^2h}
     {2(1+2\gamma bh)}
\right),
\]
provided
\[
1+2\gamma bh>0.
\]
Therefore,
\begin{equation}
\Pi_0(h)
=
bh
-
\frac1{2\gamma}\log(1+2\gamma bh)
+
\frac{\gamma a^2h}
     {2(1+2\gamma bh)}.
\label{eq:quadratic-exponential-exact}
\end{equation}

Expanding \eqref{eq:quadratic-exponential-exact} around \(h=0\) gives
\begin{equation*}
\Pi_0(h)
=
\frac{\gamma a^2}{2}h
+
\left(
\gamma b^2-\gamma^2a^2b
\right)h^2
+
o(h^2).
\label{eq:quadratic-exponential-expansion}
\end{equation*}
This agrees with \eqref{eq:quadratic-c1}--\eqref{eq:quadratic-c2},
since exponential utility satisfies
\[
\frac{u''}{u'}=-\gamma,
\qquad
\frac{u'''}{u'}=\gamma^2,
\qquad
\frac{u^{(4)}}{u'}=-\gamma^3.
\]
In particular,
\[
\frac{u^{(4)}}{u'}
-
\left(\frac{u''}{u'}\right)^3
=
0,
\]
so that the pure fourth-moment and certainty-equivalent inversion terms
cancel under constant absolute risk aversion.

\subsection{A Quadratic Payoff in the Vasicek Model}
\label{subsec:vasicek-example}

Consider the Vasicek short-rate model
\begin{equation*}
dr_t
=
\kappa(\theta-r_t)\,dt+\sigma\,dW_t,
\qquad
\kappa,\sigma>0.
\label{eq:vasicek-dynamics}
\end{equation*}
Its solution at time \(T\) is
\begin{equation*}
r_T
=
m_T+G_T,
\label{eq:vasicek-solution}
\end{equation*}
where
\[
m_T
=
\theta+(r_0-\theta)e^{-\kappa T}
\]
and
\begin{equation*}
G_T
=
\sigma
\int_0^T
e^{-\kappa(T-s)}
\,dW_s.
\label{eq:vasicek-gaussian-factor}
\end{equation*}
\(G_T\) is Gaussian with variance
\begin{equation*}
v_T
=
\mathbb E[G_T^2]
=
\sigma^2
\int_0^T e^{-2\kappa(T-s)}\,ds
=
\frac{\sigma^2}{2\kappa}
\left(
1-e^{-2\kappa T}
\right).
\label{eq:vasicek-total-variance}
\end{equation*}

We consider the payoff
\begin{equation*}
X
=
\mu+aG_T+b(G_T^2-v_T),
\label{eq:vasicek-payoff}
\end{equation*}
where \(a,b\in\mathbb R\). The linear term represents first-order
exposure to the terminal short-rate innovation, while the quadratic
term captures convexity with respect to the Gaussian rate factor. Such
a payoff may also be interpreted as a second-order approximation to the
profit and loss of an interest-rate position.

Let
\[
K_T(s)
=
\sigma e^{-\kappa(T-s)}.
\]
Then
\[
G_T
=
I_1(K_T)
\]
and
\[
G_T^2-v_T
=
I_2(K_T^{\otimes2}).
\]
Consequently,
\begin{equation*}
X
=
\mu
+
aI_1(K_T)
+
bI_2(K_T^{\otimes2}),
\label{eq:vasicek-chaos-decomposition}
\end{equation*}
which is a finite mixed first--second Wiener chaos expansion.

For \(0<h<T\), define
\begin{equation*}
G_{T,h}
=
\sigma
\int_0^h
e^{-\kappa(T-s)}
\,dW_s
=
I_1(K_T\mathbf1_{[0,h]}).
\label{eq:vasicek-revealed-factor}
\end{equation*}
Its variance is
\begin{align*}
v_h
&=
\mathbb E[G_{T,h}^2]
 =
\sigma^2
\int_0^h
e^{-2\kappa(T-s)}
\,ds
 =
\frac{\sigma^2e^{-2\kappa T}}{2\kappa}
\left(
e^{2\kappa h}-1
\right).
\label{eq:vasicek-revealed-variance}
\end{align*}

Since the Brownian increments over \([0,h]\) and \((h,T]\) are
independent,
\[
\mathbb E[G_T\mid\mathcal F_h]
=
G_{T,h}.
\]
Moreover,
\[
\mathbb E[G_T^2-v_T\mid\mathcal F_h]
=
G_{T,h}^2-v_h.
\]
It follows that the conditional payoff innovation is
\begin{equation*}
Y_h
=
\mathbb E[X\mid\mathcal F_h]-\mu
=
aG_{T,h}
+
b(G_{T,h}^2-v_h).
\label{eq:vasicek-conditional-innovation}
\end{equation*}

Writing
\[
Z_h
=
\frac{G_{T,h}}{\sqrt{v_h}}
\sim N(0,1),
\]
we obtain the Hermite representation
\begin{equation*}
Y_h
=
a\sqrt{v_h}\,H_1(Z_h)
+
bv_h\,H_2(Z_h),
\label{eq:vasicek-hermite-representation}
\end{equation*}
where
\[
H_1(z)=z,
\qquad
H_2(z)=z^2-1.
\]
Thus the Vasicek example has exactly the same Hermite structure as the
Brownian quadratic example, but with calendar time \(h\) replaced by the
variance \(v_h\) of the rate information revealed by time \(h\).

The relevant moments are
\begin{equation*}
\mathbb E[Y_h^2]
=
a^2v_h+2b^2v_h^2,
\label{eq:vasicek-second-moment}
\end{equation*}
\begin{equation*}
\mathbb E[Y_h^3]
=
6a^2bv_h^2+8b^3v_h^3,
\label{eq:vasicek-third-moment}
\end{equation*}
and
\begin{equation*}
\mathbb E[Y_h^4]
=
3a^4v_h^2+O(v_h^3).
\label{eq:vasicek-fourth-moment}
\end{equation*}
The mixed third moment is generated by the Hermite contraction
\[
\mathbb E[H_1(Z_h)^2H_2(Z_h)]
=
2.
\]

Expanding expected utility gives
\begin{equation*}
\mathbb E[u(\mu+Y_h)]
=
u(\mu)
+
A_1v_h
+
A_2v_h^2
+
o(v_h^2),
\label{eq:vasicek-utility-variance-expansion}
\end{equation*}
where
\begin{equation*}
A_1
=
\frac12u''(\mu)a^2
\label{eq:vasicek-A1}
\end{equation*}
and
\begin{equation*}
A_2
=
u''(\mu)b^2
+
u'''(\mu)a^2b
+
\frac18u^{(4)}(\mu)a^4.
\label{eq:vasicek-A2}
\end{equation*}

It follows that the dynamic risk premium admits the variance-time
expansion
\begin{equation}
\Pi_0(h)
=
C_1v_h+C_2v_h^2+o(v_h^2),
\label{eq:vasicek-premium-variance-time}
\end{equation}
where
\begin{equation*}
C_1
=
-\frac12
\frac{u''(\mu)}{u'(\mu)}
a^2
\label{eq:vasicek-C1}
\end{equation*}
and
\begin{align*}
C_2
={}&
-\frac{u''(\mu)}{u'(\mu)}b^2
-
\frac{u'''(\mu)}{u'(\mu)}a^2b
\nonumber\\
&-
\frac18
\left[
\frac{u^{(4)}(\mu)}{u'(\mu)}
-
\frac{u''(\mu)^3}{u'(\mu)^3}
\right]a^4.
\label{eq:vasicek-C2}
\end{align*}

To express this result directly in calendar time, observe that
\begin{equation}
v_h
=
\sigma^2e^{-2\kappa T}h
+
\kappa\sigma^2e^{-2\kappa T}h^2
+
o(h^2).
\label{eq:vasicek-small-h-variance}
\end{equation}
Define
\[
K_0
=
K_T(0)
=
\sigma e^{-\kappa T}.
\]
Then
\[
v_h
=
K_0^2h+\kappa K_0^2h^2+o(h^2).
\]
Substitution into
\eqref{eq:vasicek-premium-variance-time} gives
\begin{equation*}
\Pi_0(h)
=
c_1h+c_2h^2+o(h^2),
\label{eq:vasicek-calendar-expansion}
\end{equation*}
where
\begin{equation*}
c_1
=
-\frac12
\frac{u''(\mu)}{u'(\mu)}
a^2K_0^2
\label{eq:vasicek-calendar-c1}
\end{equation*}
and
\begin{align}
c_2
={}&
-\frac{\kappa}{2}
\frac{u''(\mu)}{u'(\mu)}
a^2K_0^2
 -
\frac{u''(\mu)}{u'(\mu)}
b^2K_0^4
-
\frac{u'''(\mu)}{u'(\mu)}
a^2bK_0^4
 -
\frac18
\left[
\frac{u^{(4)}(\mu)}{u'(\mu)}
-
\frac{u''(\mu)^3}{u'(\mu)^3}
\right]
a^4K_0^4.
\label{eq:vasicek-calendar-c2}
\end{align}

The first term in \eqref{eq:vasicek-calendar-c2} is generated by the
local variation of the Vasicek kernel, $K_T'(0)=\kappa K_0$. The remaining terms are the pure second-chaos, mixed prudence, and
fourth-moment contributions already identified in the general
first--second chaos expansion.

The Malliavin derivatives are
\begin{equation*}
D_sX
=
aK_T(s)+2bG_TK_T(s)
\label{eq:vasicek-first-malliavin}
\end{equation*}
and
\begin{equation*}
D_{s,t}^2X
=
2bK_T(s)K_T(t).
\label{eq:vasicek-second-malliavin}
\end{equation*}
Therefore,
\[
\mathbb E[D_0X]
=
aK_0
\]
and
\[
\mathbb E[D_0^2X]
=
2bK_0^2.
\]
All higher Malliavin derivatives vanish. Hence
\begin{align*}
\Pi_0(h)
={}&
-\frac12
\frac{u''(\mu)}{u'(\mu)}
\bigl(\mathbb E[D_0X]\bigr)^2h
 -
\Bigg\{
\frac{\kappa}{2}
\frac{u''(\mu)}{u'(\mu)}
\bigl(\mathbb E[D_0X]\bigr)^2
 \quad
+
\frac14
\frac{u''(\mu)}{u'(\mu)}
\bigl(\mathbb E[D_0^2X]\bigr)^2
\nonumber\\
& 
+
\frac12
\frac{u'''(\mu)}{u'(\mu)}
\bigl(\mathbb E[D_0X]\bigr)^2
\mathbb E[D_0^2X]
 \quad
+
\frac18
\left[
\frac{u^{(4)}(\mu)}{u'(\mu)}
-
\frac{u''(\mu)^3}{u'(\mu)^3}
\right]
\bigl(\mathbb E[D_0X]\bigr)^4
\Bigg\}h^2
+
o(h^2).
\label{eq:vasicek-malliavin-expansion}
\end{align*}

\paragraph{Exponential utility.}

For exponential utility,
\[
u(x)=-e^{-\gamma x},
\qquad \gamma>0,
\]
the dynamic risk premium is available in closed form:
\begin{equation*}
\Pi_0(h)
=
\frac1\gamma
\log\mathbb E[e^{-\gamma Y_h}].
\label{eq:vasicek-exponential-definition}
\end{equation*}
Since \(G_{T,h}\sim N(0,v_h)\),
\begin{align*}
\mathbb E[e^{-\gamma Y_h}]
={}&
e^{\gamma bv_h}
(1+2\gamma bv_h)^{-1/2}
\nonumber\\
&\times
\exp
\left(
\frac{\gamma^2a^2v_h}
     {2(1+2\gamma bv_h)}
\right),
\label{eq:vasicek-exponential-mgf}
\end{align*}
provided
\[
1+2\gamma bv_h>0.
\]
Consequently,
\begin{equation*}
\Pi_0(h)
=
bv_h
-
\frac1{2\gamma}
\log(1+2\gamma bv_h)
+
\frac{\gamma a^2v_h}
     {2(1+2\gamma bv_h)}.
\label{eq:vasicek-exponential-exact}
\end{equation*}

Its small-\(v_h\) expansion is
\begin{equation*}
\Pi_0(h)
=
\frac{\gamma a^2}{2}v_h
+
\left(
\gamma b^2-\gamma^2a^2b
\right)v_h^2
+
o(v_h^2).
\label{eq:vasicek-exponential-variance-expansion}
\end{equation*}
Using \eqref{eq:vasicek-small-h-variance}, this becomes
\begin{align*}
\Pi_0(h)
={}&
\frac{\gamma a^2K_0^2}{2}h
 +
\left[
\frac{\gamma\kappa a^2K_0^2}{2}
+
\left(
\gamma b^2-\gamma^2a^2b
\right)K_0^4
\right]h^2
+
o(h^2).
\label{eq:vasicek-exponential-calendar-expansion}
\end{align*}

The exact formula provides an independent verification of the general
mixed-chaos expansion. It also illustrates how mean reversion affects
the dynamic risk premium through the information variance \(v_h\).
For fixed \(T\), the initial exposure of the terminal rate to new
Brownian information is $K_0=\sigma e^{-\kappa T}$. Consequently, stronger mean reversion reduces the leading risk premium
by attenuating the effect of shocks occurring near the initial time on
the terminal short rate.

\section{Asymptotic Validity of the Arrow--Pratt Approximation for Wiener Functionals}
\label{sec:Asymptotic validity}

The examples developed in the previous section provide a useful
contrast with the counterexample presented in Proposition~1.
There we constructed a sequence of small risks for which the
Arrow--Pratt approximation fails to be asymptotically correct, even
though the risks converge to zero in every $L^p$.
The Wiener-chaos framework leads to a fundamentally different
conclusion.

The key distinction lies in the mechanism by which the uncertainty
vanishes. In Proposition~1, the variance tends to zero because the
probability of a finite loss converges to zero, while the size of the
loss itself remains unchanged. Consequently, the limiting risk premium
depends on the finite utility difference
\[
u(w)-u(w-a),
\]
rather than solely on the local curvature of the utility function.

By contrast, the Wiener-chaos framework describes local perturbations
generated by Brownian motion. The conditional payoff increment
\[
Y_h
=
E[X|\mathcal F_h]-E[X]
\]
converges locally to zero, and its moments satisfy the
hypercontractive estimates associated with finite Wiener chaoses.
This additional regularity restores the asymptotic validity of the
Arrow--Pratt approximation.
 
To illustrate this phenomenon, consider first the quadratic Gaussian
example
\[
Y_h
=
aW_h+b(W_h^2-h),
\]
together with exponential utility
\[
u(x)
=
-e^{-\gamma x},
\qquad
\gamma>0.
\]

Section~\ref{subsec:quadratic-gaussian-example} showed that the dynamic risk premium is given exactly by
\[
\Pi_0(h)
=
bh
-
\frac{1}{2\gamma}
\log(1+2\gamma bh)
+
\frac{\gamma a^2h}
     {2(1+2\gamma bh)}.
\]

Expanding around $h=0$ gives
\[
\Pi_0(h)
=
\frac{\gamma a^2}{2}h
+
\left(
\gamma b^2
-
\gamma^2a^2b
\right)h^2
+
O(h^3).
\]

On the other hand,
\[
\operatorname{Var}(Y_h)
=
a^2h+2b^2h^2,
\]
so that the classical Arrow--Pratt approximation becomes
\[
\Pi_{AP}(h)
=
\frac{\gamma}{2}
\operatorname{Var}(Y_h)
=
\frac{\gamma a^2}{2}h
+
\gamma b^2h^2.
\]

Consequently,
\[
\Pi_0(h)-\Pi_{AP}(h)
=
-\gamma^2a^2bh^2
+
O(h^3).
\]

The leading Arrow--Pratt term is therefore recovered exactly, while the
first correction is generated by the interaction between the first and
second Wiener chaoses.

Exactly the same phenomenon occurs for the Vasicek model. For the payoff
\[
X
=
\mu+aG_T+b(G_T^2-v_T),
\]
the conditional payoff increment is
\[
Y_h
=
aG_{T,h}
+
b(G_{T,h}^2-v_h),
\]
where
\[
v_h
=
\operatorname{Var}(G_{T,h}).
\]

The exact exponential-utility expansion obtained in Section~\ref{subsec:vasicek-example} is
\[
\Pi_0(h)
=
\frac{\gamma a^2}{2}v_h
+
\left(
\gamma b^2
-
\gamma^2a^2b
\right)
v_h^2
+
O(v_h^3).
\]

Since
\[
\operatorname{Var}(Y_h)
=
a^2v_h
+
2b^2v_h^2,
\]
the corresponding Arrow--Pratt approximation is
\[
\Pi_{AP}(h)
=
\frac{\gamma}{2}
\operatorname{Var}(Y_h)
=
\frac{\gamma a^2}{2}v_h
+
\gamma b^2v_h^2.
\]

Hence
\[
\Pi_0(h)-\Pi_{AP}(h)
=
-\gamma^2a^2bv_h^2
+
O(v_h^3).
\]

The only effect of the Vasicek dynamics is that calendar time is
replaced by the revealed-information variance
\[
v_h
=
\frac{\sigma^2e^{-2\kappa T}}
{2\kappa}
\left(
e^{2\kappa h}-1
\right),
\]
while the structure of the higher-order correction remains unchanged.

The previous examples suggest that the behaviour observed above is not
accidental but is a structural property of finite Wiener-chaos
expansions.

\begin{prop}[Asymptotic validity of the Arrow--Pratt approximation]
\label{prop:arrow-pratt-asymptotic}
Suppose that
\[
X
=
\mu
+
\sum_{n=1}^{N}
I_n(f_n),
\]
where the kernels are continuous in a neighbourhood of the initial
diagonal, and let
\[
m
=
\min
\{
n:
f_n\neq0
\}
\]
denote the lowest non-vanishing Wiener chaos.

Then the dynamic risk premium satisfies
\[
\Pi_0(h)
=
\Pi_{AP}(h)
+
O(h^{m+1}),
\]
where
\[
\Pi_{AP}(h)
=
-\frac12
\frac{u''(\mu)}
     {u'(\mu)}
\operatorname{Var}
(E[X|\mathcal F_h]).
\]



\end{prop}

\begin{proof}

By Theorem~4,
the conditional payoff increment satisfies
\[
Y_h
=
O_{L^2}(h^{m/2}).
\]

Hypercontractivity implies
\[
\|Y_h\|_p
=
O(h^{m/2})
\]
for every finite $p$.

Consequently,
\[
E[u(\mu+Y_h)]
=
u(\mu)
+
\frac12
u''(\mu)
E[Y_h^2]
+
O(h^{m+1}),
\]
since all moments of order three and higher contribute only at order
$h^{m+1}$ or smaller.

Similarly,
\[
u(\mu-\Pi_0(h))
=
u(\mu)
-
u'(\mu)\Pi_0(h)
+
O(h^{m+1}),
\]
because
\[
\Pi_0(h)
=
O(h^m).
\]

Comparison of the two expansions yields
\[
\Pi_0(h)
=
-
\frac12
\frac{u''(\mu)}
     {u'(\mu)}
E[Y_h^2]
+
O(h^{m+1}),
\]
which is precisely the Arrow--Pratt approximation together with the
stated remainder estimate.
\end{proof}

Proposition~\ref{prop:arrow-pratt-asymptotic} provides a precise mathematical characterization of the range
of validity of the classical Arrow--Pratt approximation.
Without further assumptions, Proposition~1 shows that the approximation
may fail even for sequences of risks converging to zero in every
$L^p$.
For regular Wiener functionals, however, the additional structure
provided by Malliavin differentiability and the Wiener-chaos
decomposition restores the asymptotic validity of the Arrow--Pratt
formula.

The higher-order terms are no longer arbitrary. They are completely
determined by the interaction between different Wiener chaos orders,
their associated contraction operators, and the higher derivatives of
the utility function. Consequently, the Wiener-chaos decomposition not
only yields a hierarchy of time scales at which uncertainty enters the
dynamic certainty equivalent, but also provides a complete hierarchy of
dynamic Arrow--Pratt coefficients extending the classical theory to
arbitrary order.

\section{Conclusion}

This paper has developed a new dynamic framework for the analysis of certainty
equivalents and risk premia by combining expected utility theory with
Malliavin calculus and Wiener chaos analysis. Our starting point was a
re-examination of the classical Arrow--Pratt approximation. We showed that,
contrary to a common interpretation, the Arrow--Pratt formula is not
asymptotically valid for arbitrary sequences of vanishing risks. In
particular, small variance alone does not guarantee asymptotic accuracy,
highlighting the need for additional structural assumptions on the underlying
uncertainty.

Motivated by this observation, we reformulated the certainty-equivalent
problem dynamically by considering the progressive revelation of uncertainty
through a Brownian filtration. This leads naturally to an infinitesimal
dynamic risk premium, which recovers the classical Arrow--Pratt coefficient as
its leading-order term and admits an explicit representation in terms of the
Clark--Ocone integrand and the Malliavin derivative of the terminal payoff.

Building on this dynamic viewpoint, we derived a hierarchy of higher-order
dynamic risk premia using the Wiener chaos decomposition. For pure Wiener
chaos components, the order of the chaos determines the first order at which
the corresponding uncertainty contributes to the certainty equivalent, giving
rise to a natural hierarchy indexed by the chaos order. The associated
coefficients admit explicit expressions in terms of Malliavin derivatives,
providing a direct stochastic-analytic interpretation of higher-order dynamic
risk premia.

For mixed Wiener chaos expansions, a richer structure emerges. Interactions
between different chaos orders generate higher-order preference measures,
including prudence and temperance, while Bell polynomial representations
describe the nonlinear transformation from expected utility to certainty
equivalents. This establishes, to our knowledge, the first systematic link
between the classical hierarchy of higher-order risk preferences in economics
and the orthogonal decomposition of uncertainty provided by Wiener chaos
analysis.

The examples considered in this paper further demonstrate that a broad class
of regular Wiener functionals satisfies the conditions under which the
classical Arrow--Pratt approximation is recovered asymptotically as the
leading-order term, while the higher-order corrections are entirely determined
by the local Wiener chaos structure. In this sense, the present framework not
only explains when the classical approximation is valid, but also quantifies
its higher-order corrections in a mathematically transparent way.

Several directions for future research appear promising. From a theoretical perspective, it would be natural to extend the analysis to
more general semimartingale settings, including jump processes and Lévy-driven
models, where the Wiener chaos expansion is replaced by the corresponding
Poisson or Lévy chaos decomposition. It
would also be of interest to investigate analogous hierarchies under
alternative preference models, such as recursive or ambiguity-sensitive
utilities. On the applied side, the dynamic risk-premium expansions developed
here may provide new approximation techniques for utility-based pricing,
portfolio optimisation, indifference valuation and stochastic control in
high-dimensional continuous-time models.

More broadly, the results suggest that Malliavin calculus and Wiener chaos
analysis provide considerably more than computational tools for stochastic
analysis. They offer a natural analytical language for studying dynamic
certainty equivalents and higher-order risk preferences, thereby establishing
new connections between expected utility theory, stochastic analysis and
mathematical finance.

\end{document}